%% file: main_arxiv20231008.tex
\documentclass[unicode,runningheads]{llncs}
\usepackage[crop]{llncsconf} 
\usepackage[backend=biber,url=false,arxiv=abs,maxbibnames=100,isbn=false,sortcites=true,style=lncs]{biblatex}
\usepackage{graphicx}
\usepackage{color}
\usepackage{mathtools}
\usepackage{todonotes}
\usepackage{amsmath,amsfonts,amssymb}
\usepackage{empheq}
\usepackage{ascmac}
\usepackage{comment}

\usepackage{orcidlink}
\renewcommand{\orcidID}[1]{\orcidlink{#1}}
\usepackage[capitalize]{cleveref}
\let\doendproof\endproof
\renewcommand\endproof{~\hfill$\qed$\doendproof}

\newcommand{\C}{\mathcal{C}}

\newcommand{\vc}{\tau}

\newcommand{\PoS}{{\rm{PoS} }}

\newcommand{\sw}{{\mathtt{sw}}}
\newcommand{\uw}{{\mathtt{uw}}}
\newcommand{\ew}{{\mathtt{ew}}}
\newcommand{\PP}{{\mathcal P}}
\newcommand{\DP}{{\rm{DP}}}

\newcommand{\diam}{\mathrm{diam}}
\newcommand{\dist}{\mathrm{dist}}

\newcommand{\cC}{\mathcal{C}}
\newcommand{\cT}{\mathcal{T}}
\newcommand{\cX}{\mathcal{X}}

\newcommand{\iso}{\mathtt{iso}}
\newcommand{\cl}{\mathtt{cl}}
\newcommand{\stc}{\mathtt{sc}}
\newcommand{\stcu}{\mathtt{scu}}
\newcommand{\stl}{\mathtt{sl}}

\newcommand{\typeB}{\mathtt{type}_B}
\newcommand{\typec}{\mathtt{type}_c}

\newcommand{\fin}{\mathtt{fin}}

\begin{document}
%
\title{Maximizing Utilitarian and Egalitarian Welfare of Fractional Hedonic Games \\
on Tree-like Graphs~\thanks{%
Partially supported
by JSPS KAKENHI Grant Numbers 
JP20H05967, 
JP21H05852, 
JP21K17707, 
JP21K19765, 
JP22H00513, 
and JP23H04388. 
}
}
\titlerunning{Utilitarian and Egalitarian Welfare of Fractional Hedonic Games}
%
\author{
Tesshu Hanaka\inst{1}\orcidID{0000-0001-6943-856X} \and
Airi Ikeyama\inst{2} \and
Hirotaka Ono\inst{2}\orcidID{0000-0003-0845-3947}
}
\authorrunning{Hanaka et al.}
%
\institute{Kyushu University, Fukuoka, Japan \email{hanaka@inf.kyushu-u.ac.jp} \and
Nagoya University, Nagoya, Japan \email{ikeyama.airi.f4@s.mail.nagoya-u.ac.jp,ono@nagoya-u.jp}}
%
%
\maketitle
%
\begin{abstract}
Fractional hedonic games are coalition formation games where a player's utility is determined by the average value they assign to the members of their coalition. These games are a variation of graph hedonic games, which are a class of coalition formation games that can be succinctly represented. Due to their applicability in network clustering and their relationship to graph hedonic games, fractional hedonic games have been extensively studied from various perspectives. However, finding welfare-maximizing partitions in fractional hedonic games is a challenging task due to the nonlinearity of utilities. In fact, it has been proven to be NP-hard and can be solved in polynomial time only for a limited number of graph classes, such as trees. 
This paper presents (pseudo)polynomial-time algorithms to compute welfare-maximizing partitions in fractional hedonic games on tree-like graphs. 
We consider two types of social welfare measures: utilitarian and egalitarian. Tree-like graphs refer to graphs with bounded treewidth and block graphs. 
A hardness result is provided, demonstrating that the pseudopolynomial-time solvability is the best possible under the assumption P$\neq$NP.
\keywords{Fractional hedonic game  \and treewidth \and block graphs.}
\end{abstract}
\input{chapter01.tex}


\input{chapter02.tex}




\input{chapter04_ver3.tex}


\input{chapter05.tex}

\printbibliography
\end{document}

%% file: chapter01.tex
\section{Introduction}
\subsection{Definition and motivation}
The hedonic game \cite{Dreze1980} is a game of modeling coalition formation based on individual preferences. 
Graphical variants of hedonic games have been considered to express the preferences succinctly. In this work, we deal with a variant of the graphical hedonic games called a \emph{fractional hedonic game} (FHG), a subclass of hedonic games in which each agent's utility is the average of valuations over the other agents in the belonging coalition.
Similarly to ordinary graphical hedonic games, fractional hedonic games are coalition formation games on a graph, in which vertices represent agents and the weight of edge $(i, j)$ denotes the value that agent $i$ has for agent $j$.
Although fractional hedonic games in the most general setting are defined on weighted directed graphs, simpler versions of fractional hedonic games are well studied~\cite{Aziz2019}. For example, fractional hedonic games are said to be \emph{symmetric} and are represented on undirected graphs when all pairs of two agents are equally friendly. Furthermore, fractional hedonic games are said to be \emph{simple} and are represented on unweighted graphs when all of the edge weights are 1.
This work deals with fractional hedonic games on undirected graphs as \cite{Aziz2019}.

In fractional hedonic games, a coalition structure is represented by a partition of the vertices, where each set represents a coalition. 
Given a coalition structure, the utility of each agent is defined as the average weights of its incident edges, as explained above. This definition implies that if two coalitions contain an identical set of agents (vertices) adjacent to agent $v$, the sums of the edge weights incident to $v$ are equal, but the smaller coalition is more desirable for $v$. 
Such a property is suitable for finding a partition into dense subgraphs, which is why it is used for network clustering. 

Under this definition of individual utility, two significant measures are well considered in welfare maximization: utilitarian or egalitarian social welfare. The former measure is the sum of the utilities of all agents. The latter is the minimum utility among the utilities of all agents. Hence, we call the partitions that maximize utilitarian and egalitarian welfare the maximum utilitarian welfare coalition structure (partition) and the maximum egalitarian welfare coalition structure (or partition), respectively. 


Although many papers have already studied the computational complexity of finding the maximum utilitarian and egalitarian welfare coalition structure, positive results are few; even for restricted classes of graphs, they are NP-hard, and polynomial-time solvable classes are very restricted as summarized in Table~\ref{table:generalresult}. 
Noteworthy, the complexity of the problem for bounded treewidth remained open~\cite{DBLP:conf/swat/BazganCC22}. Such unwieldiness might be due to the nonlinearity of the objective function of the problem.

Therefore, this paper tries to enlarge solvable classes of graphs for the problems. In particular, we focus on block graphs and design a polynomial-time algorithm for computing the maximum utilitarian welfare coalition structure. Furthermore, we also focus on graphs with bounded vertex cover numbers and treewidth, for which we can design (pseudo)polynomial-time algorithms for finding the maximum utilitarian and egalitarian welfare coalition structures, which resolves the open problem left in \cite{DBLP:conf/swat/BazganCC22}. 
At the same time, finding the maximum egalitarian welfare partition is shown to be weakly NP-hard even for graphs with vertex cover number 4. The detailed results obtained in this paper are summarized in Our contribution.


\subsection{Our contribution}
This paper first shows a polynomial-time algorithm for computing the maximum utilitarian welfare partition on block graphs (Theorem~\ref{thm:block:sw:opt}). 
We see how a coalition structure forms in the maximum utilitarian welfare coalition structure on a block graph; by computing average utilities elaborately, we can show that there is a maximum utilitarian welfare coalition structure on a block graph in which every coalition forms a clique or a star.   
By utilizing the characterization, we design a dynamic programming-based algorithm that runs along a tree structure of a block graph.

We then focus on the complexity of maximizing utilitarian and egalitarian welfare for well-known graph parameters: treewidth and vertex cover number, which indicate how tree-like or star-like a graph is, respectively (Table~\ref{table:result}). 
For the utilitarian welfare maximization, we give an $(nW)^{O(\omega)}$-time algorithm and an $n^{O(\tau)}$-time algorithm where $\omega$ is the treewidth of $G$, $\tau$ is the vertex cover number of $G$, and $W$ is the maximum absolute weight of edges. 
This resolves the open question left in \cite{DBLP:conf/swat/BazganCC22} of whether the utilitarian welfare maximization on unweighted fractional hedonic games (equivalently, \textsc{Dense Graph Partitioning}) can be solved in polynomial time on bounded treewidth graphs.
We mention that it remains open whether there exists a polynomial-time algorithm on \emph{weighted} bounded treewidth graphs though we can design the one on bounded vertex cover number graphs.

For the egalitarian welfare maximization, we also give an $(nW)^{O(\omega)}$-time algorithm for treewidth $\omega$.
We then show that in contrast to utilitarian welfare maximization, egalitarian welfare maximization is NP-hard on bounded vertex cover number graphs; it implies that the pseudopolynomial-time solvability is best possible under the assumption P$\neq$NP.

\begin{table}[tb]
\begin{center}
\caption{Complexity of the utilitarian and egalitarian welfare maximization on graph classes.}
\label{table:generalresult}
\begin{tabular}{l|l||l}
\hline
objective & graph class & complexity  \\
\hline
\hline
utilitarian  & general, unweighted & NP-hard \cite{DBLP:conf/ijcai/AzizGGMT15,Bilo2018} \\
 & cubic graphs, unweighted &  NP-hard \cite{DBLP:conf/swat/BazganCC22} \\
& $\delta\ge n-3$, unweighted & P \cite{DBLP:conf/swat/BazganCC22}  \\
& bipartite, unweighted &  NP-hard \cite{DBLP:conf/swat/BazganCC22} \\
 & block, unweighted & P [Thm.~\ref{thm:block:sw:opt}] \\
 & tree, unweighted & P~\cite{Bilo2018} \\  \hline
 egalitarian & general, unweighted & NP-hard \cite{DBLP:conf/ijcai/AzizGGMT15} \\
 \hline
\end{tabular}
\end{center}
\end{table}

\begin{table}[tb]

\begin{center}
\caption{Complexity of the utilitarian and egalitarian welfare maximization with graph parameters.}
\label{table:result}
\begin{tabular}{l||l|l|l}
\hline
  objective & parameter & unweighted & weighted   \\
\hline
\hline
 utilitarian & treewidth & $n^{O(\omega)}$ [Thm.~\ref{thm:treewidth:weighted}] & $(nW)^{O(\omega)}$ [Thm.~\ref{thm:treewidth:weighted}]    \\ 
  & vertex cover number & $n^{O(\vc)}$ [Thm.~\ref{thm:vc:weighted}] &  $n^{O(\vc)}$ [Thm.~\ref{thm:vc:weighted}] \\ \hline
   
    egalitarian   & treewidth & $n^{O(\omega)}$ [Thm.~\ref{thm:max-min:weighted}] & $(nW)^{O(\omega)}$ [Thm.~\ref{thm:max-min:weighted}]    \\
   & vertex cover number & $n^{O(\vc)}$ [Cor.~\ref{cor:Egal:unweighted:vc}] &  paraNP-hard [Thm.~\ref{thm:Egal:vc}] \\

 \hline
\end{tabular}
    \end{center}
\end{table}

\subsection{Related work}
Fractional hedonic games have been studied from several aspects, including complexity, algorithm, and stability. 

\subsubsection{Complexity and algorithm}
From the algorithmic point of view, a coalition structure maximizing utilitarian welfare or egalitarian welfare is NP-hard to compute~\cite{DBLP:conf/ijcai/AzizGGMT15,Bilo2018}.
On the other hand, computing a coalition structure with maximum utilitarian welfare is shown to be solvable in polynomial time only for a few graph classes, such as trees~\cite{Bilo2018}. 
Furthermore, a problem equivalent to computing the maximum utilitarian welfare coalition structure in fractional hedonic games is studied under a different name; Dense graph partitioning~\cite{DARLAY20122389,DBLP:conf/swat/BazganCC22}. 
Dense graph partitioning is the problem of finding a partition with maximum density for a given graph. From the study on Dense graph partitioning, it is known that computing the maximum utilitarian welfare coalition structure is NP-hard even for cubic graphs~\cite{DBLP:conf/swat/BazganCC22}. 
Table~\ref{table:generalresult} shows the complexity of computing the maximum utilitarian welfare coalition structure and computing the maximum egalitarian welfare coalition structure. $\delta$ denotes the minimum degree of the input graph. 

As a variant of fractional hedonic games, fractional hedonic games with the restriction that there is a specific upper bound $k$ on the number of coalitions that can be formed have also been studied~\cite{li2021fractional}.
For all fixed $k\ge2$, it remains NP-hard to find a maximum utilitarian welfare coalition structure with $k$ coalition for undirected unweighted graphs~\cite{li2021fractional}. For undirected unweighted trees, we can compute a maximum utilitarian welfare coalition structure with $k$ coalition in time $O(nk)$~\cite{li2021fractional}. 
Futhermore, as a variant of fractional hedonic games, social distance games~\cite{Brnzei2011SocialDG} have also been studied in which each agent gets benefits from all agents belonging to the same coalition, not only from adjacent agents as in fractional hedonic games. 
It is known that to compute a coalition structure that maximizes the sum of the utility in social distance games is NP-hard~\cite{Brnzei2011SocialDG}. 
\subsubsection{Stability}
A motivation to study computing the maximum social welfare is that it is often used to measure the goodness of a coalition structure satisfying a certain property. One of the most well-used properties is stability.  

The coalition structure is said to be stable when no agent increases its utility by moving to another coalition. A stable coalition structure can be considered a state in which each agent is satisfied with the coalition structure. 
On the other hand, the maximum utilitarian welfare coalition structure and the maximum egalitarian coalition structure, which are generally considered socially desirable coalition structures, are not necessarily stable. 
The Price of Stability (\PoS)~\cite{Anshelevich2008} is a measure of the gap between the stable coalition structure and the maximum utilitarian welfare coalition. The price of stability is defined as the total utility of the maximum utilitarian welfare coalition structure divided by the total utility of the coalition structure that maximizes utilitarian welfare in stable coalition structures. It represents at least how much utilitarian welfare is affected to achieve stability.
A lower bound of $2.2247$ is known for the \PoS in the general graph~\cite{Kaklamanis2020}. On the other hand, only $n-1$ is shown for the upper bound~\cite{Bilo2018}. Since the gap between the upper and lower bounds known is very large, many researchers try to fill the gap by restricting graph classes. 
For example, it is known that for graphs with girth (shortest closed path length) greater than 4, 
a lower bound $1.0025$ and an upper bound $2.5714$ of \PoS are known~\cite{Bilo2018}. 
For bipartite graphs, 
a lower bound $1.0025$ and an upper bound $1.0294$ are known~\cite{Bilo2018}.
It is also known that for graphs with an even larger girth, i.e., with a girth at lease 5, \PoS is 1 and there is always a stable and optimal solution~\cite{Kaklamanis2020}.

As with other stability concepts, Pareto-Optimality in fractional hedonic games has been studied~\cite{10.5555/3398761.3398791, 10.5555/3398761.3399119}. In simple symmetric fractional hedonic games in which they are equally friendly to each other, and all of the edge weight is 1, finding a Pareto-optimal partition can be done in polynomial time~\cite{10.5555/3398761.3398791}. 
Recently, the notion of popularity is getting a lot of attention. A partition is popular if there is no other partition in which more agents get better than worse. 
Computing popular, strongly popular, and mixed popular partitions in symmetric fractional hedonic games are studied in~\cite{brandt2022finding}. 

\medskip 


The rest of the paper is organized as follows. Section \ref{sec:pre} introduces definitions, notations, and sees the basic properties of fractional hedonic games. Section \ref{chap:block:characterization} presents a characterization of a coalition structure with the maximum utility on a block graph. Section \ref{chap:block:algorithm} presents an algorithm to compute a coalition structure with the maximum utility on a block graph. Section \ref{chap:XP} is for graphs with bounded treewidth or vertex cover number. We present (pseudo)polynomial-time algorithms and give an NP-hardness proof.

%% file: chapter02.tex
\section{Preliminaries}\label{sec:pre}

\subsection{Definitions, terminologies, and notation}
Let $G=(V,E)$ be an undirected graph. We denote by $G[C]=(C,E(C))$ the subgraph of $G$ induced by $C\subseteq V$ and by $N_G(v)$ the set of neighbors of $v$ in $G$.
The degree $d_G(v)$ of $v$ is defined by $d_G(v)=|N_G(v)|$. We denote by $\Delta_G$ the maximum degree of $G$.
The \emph{distance} of $u$ and $v$ is defined by the length of the shortest path between them and denoted by $\dist(u,v)$.
The \emph{diameter} of $G$ denote with $\diam(G)$ and  $\diam(G)=\max_{u,v\in V} \dist(u,v)$. 
For a connected graph $G$, a vertex $v$ is called a \emph{cut vertex} if the graph obtained by deleting $v$ is disconnected.
Similarly, for a connected graph $G$, an edge $e$ is called a \emph{bridge} if the graph obtained by deleting $e$ is disconnected.

An \emph{isolated vertex} is a vertex with degree zero. A vertex set $K$ is called a \emph{clique} if $G[K]$ is a complete graph. A vertex set $I$ is called an \emph{independent set} if no two vertices in $I$ are adjacent. For a vertex set $V=\{v_1,\ldots,v_n\}$, a graph whose edge set is $E=\{\{v_1,v_j\} \mid 2\le j \le n\}$ is called a \emph{star}. The vertex $v_1$ is called the center of the star, and other vertices are called leaves. 


A \emph{block graph} is a graph whose every biconnected component is a clique. 

\subsection{Graph parameters}
A vertex set  $S\subseteq V(G)$ is a \emph{vertex cover} if every edge has at least one endpoint in $S$. The \emph{vertex cover number} $\vc(G)$ is the size of a minimum vertex cover.

\emph{Treewidth} is a parameter that represents how close a graph is to a tree.  It is defined by the minimum \emph{width} among all the \emph{tree decomposition} of $G$ as follows. 

\begin{definition}
A tree decomposition of a graph $G$ is a pair $\mathcal{T}=(T,\{X_t\}_{t\in V(T)})$ such that each bag $X_t\subseteq V(G)$ is associated with node $t$ of a tree $T$, which satisfies the following condition:
\begin{itemize}
 \item 
 $\bigcup_{t\in V(T)}X_t=V(G)$,
 \item for any edge $\{u,v\}\in E(G)$, there exists a node $t$ of $T$ such that $\{u,v\}\subseteq X_t$, and 
 \item for any vertex $u\in V(G)$, the subtree of $T$ induced by $\{t\in V(T)\mid u\in X_t\}$ is connected.
\end{itemize}

The width of a tree decomposition $\mathcal{T}$ is defined by $\max_{t\in V(T)}|X_t|-1$. The treewidth $\omega(G)$ of $G$ is the minimum width among all the tree decompositions of $G$.
\end{definition}

A tree decomposition $\mathcal{T}=(T,\{X_t\}_{t\in V(T)})$ of $G$ is \emph{nice} if it satisfies the following conditions:
\begin{enumerate}
\item $T$ is rooted at a designated node $r\in V(T)$ and $|X_r|=0$,
\item for every leaf $l$ of $T$, $|X_l|=0$, and
\item each internal node has one of the following types:
\begin{description}
    \item[(\textbf{Introduce node})] A node $t$ has only one child $t'$ and it satisfies $X_t=X_{t'}\cup\{v\}$ for some vertex $v\notin X_{t'}$.
    \item[(\textbf{Forget node})] A node $t$ has only one child $t'$ and it satisfies $X_t=X_{t'}\setminus\{v\}$ for some vertex $v\in X_{t'}$.
    \item[(\textbf{Join node})] A node $t$ has exactly two children $t_1,t_2$ and it satisfies $X_t=X_{t_1}=X_{t_2}$.
\end{description}
\end{enumerate}

For a node $t\in V(T)$, we denote by $G_t=(V_t,E_t)$  the subgraph of $G$ induced by the union of bags of the subtree of $\mathcal{T}$ rooted by $t$.
Note that $G_r = G$.
It is well-known that given a tree decomposition of width at most $\omega$ in $G$, one
can compute a nice tree decomposition of width at most $\omega$ having at most $O(\omega|V|)$ nodes in polynomial time \cite{ParameterizedAlgorithms}.

\subsection{Fractional Hedonic Game}
A fractional hedonic game is defined on a weighted and directed graph $G=(V,E,w)$, whose weights represent preferences. Without loss of generality, we suppose no edge has weight $0$. 

In this paper, we consider a \emph{symmetric} fractional hedonic game, which is defined on an undirected graph. 
If all of the edge weight is 1, a fractional hedonic game is called \emph{simple} and it is defined on an unweighted graph.

A partition $\C$ of $V$ is called a \emph{coalition structure}. A vertex set $C\in \C$ is called a coalition.
The utility $U(v,C)$ of a vertex $v$ that belongs to a coalition $C$ is defined by the sum of edge weights of neighbors of $v$ in $C$  divided by $|C|$, i.e., 
\begin{align*}
    U(v,C)=\frac{\sum_{u\in N_{G[C]}(v)}w_{uv}}{|C|}.
\end{align*}
The \emph{utilitarian welfare} $\uw(C)$ of $C\in \C$ and the one $\uw(\C)$ of a coalition structure $\C$ are defined by the sum of the utility of each $v\in C$ and the sum of the utility of each vertex in $G$, respectively:
\begin{align*}
    &\uw(C)=\sum_{v\in C}U(v,C),\\
    &\uw(\C)=\sum_{C\in \C}\uw(C)=\sum_{C\in \C}\sum_{v\in C}U(v,C).
\end{align*}
Next, the \emph{egalitarian welfare} $\ew(C)$ of $C\in \C$ and the one $\ew(\C)$ of a coalition structure $\C$ are defined by the minimum utility of $v\in C$ and the minimum utility of $v\in V$ under $\C$, respectively:
\begin{align*}
    &\ew(C)=\min_{v\in C}U(v,C), \\
    &\ew(\C)=\min_{v\in C, C\in \C}U(v,C).
\end{align*}

In this paper, we consider the problem to find a maximum utilitarian welfare coalition structure and  the problem to find a maximum egalitarian welfare coalition structure.
We use $C^*$ as an optimal coalition structure.


For fractional hedonic games, the following basic properties hold.
\begin{property}[\cite{Bilo2018}]\label{prop:Bilo}
In a symmetric fractional hedonic game,  $\uw(C)=2|E(C)|/|C|$ holds for any coalition $C$.
\end{property}


\begin{property}\label{property:clique}
In a simple symmetric fractional hedonic game, if a coalition $C$ forms a clique of size $k$, $\uw(C)=k-1$.
\end{property}


\begin{property}\label{property:star}
In a simple symmetric fractional hedonic game, if a coalition $C$ forms a star of size $k$, $\uw(C)={2(k-1)}/{k}$.
\end{property}

%% file: chapter04_ver3.tex
\section{Maximizing utilitarian welfare on block graphs: Characterization}\label{chap:block:characterization}
In this section, we characterize an optimal coalition structure on block graphs, and show that there exists a maximum utilitarian welfare coalition structure on the block graph in which each coalition induces a clique or a star. 
This characterization is used for designing a polynomial-time algorithm on block graphs.
\begin{theorem}\label{thm:block:opt}
There exists a maximum utilitarian welfare coalition structure on the block graph in which each coalition induces a clique or a star. 
\end{theorem}
To show this, we show the following two lemmas.
\begin{lemma}\label{lem:opt:diam2}
Let $C$ be a coalition such that the diameter of $G[C]$ is 2. Then  $C$ can be partitioned into cliques or stars without decreasing utilitarian welfare.
\end{lemma}

\begin{proof}
We first show that a coalition $C$ that induces a block graph having diameter 2 can be partitioned into cliques or stars without decreasing utilitarian welfare.
The subgraph $G[C]$ induced by a coalition $C$ having diameter 2 contains exactly one cut vertex $v$ and each maximal clique in $G[C]$ contains $v$.
Let $C_x$ be the largest clique in $G[C]$ and its size $x$. When $x=2$, $G[C]$ is a star. Thus, we consider the case $x\ge 3$ in the following. 
Let $C'=C_x\setminus\{v\}$ and $C^{\prime\prime}=C\setminus C'$. We show that $\uw(C')+\uw(C^{\prime\prime}) \ge \uw(C)$.

Let $y$ be the number of vertices and $z$ be the number of edges in  $G[C^{\prime\prime}]$. Then the utilitarian welfare of $C$ is $\uw(C)=2((x(x-1)/2)+z)/(x+y-1)=(x(x-1)+2z)/(x+y-1)$. On the other hand, $\uw(C')=(x-1)-1=x-2$ and $\uw(C^{\prime\prime})=2z/y$. Thus,  \begin{align}\label{align:opt:diam2}(\uw(C')+\uw(C^{\prime\prime}))-\uw(C) = (y^2(x-2)+2(x-1)(z-y))/y(x+y-1).\end{align}
If $G[C^{\prime\prime}]$ contains a clique of size at least 3, then $z-y\ge 0$. Since $x\ge 3$, we have $(\uw(C')+\uw(C^{\prime\prime}))\ge \uw(C)$. 
Otherwise, $G[C^{\prime\prime}]$ is a star. By Property \ref{property:star}, $\uw(C^{\prime\prime}) = 2(y-1)/y$.
Thus, $(\uw(C')+\uw(C^{\prime\prime}))-\uw(C) = x-2 +  2(y-1)/y - 2((x(x-1)/2)+y-1)/(x+y-1)
=x^2-2x +xy-2y -2x +2 +2(x+y-1) - 2(x+y-1)/y -x(x-1) + 2y -2 =  -x +xy  + 2y- 2(x-1)/y -4 \ge 1$ if $x\ge 3$ and $y\ge 2$.
For both cases, we can partition $C$ into two coalitions without decreasing utilitarian welfare. By repeating the partitioning procedure, we can partition $C$ into stars and cliques.
\end{proof}



Then we prove that a coalition $C$ with diameter at least 3 can be partitioned into stars and cliques as in Lemma~\ref{lem:opt:diam2}. This implies the existence of an optimal coalition structure consisting of stars and cliques.
\begin{lemma}\label{lem:opt:diam3}
Let $C$ be a coalition such that the diameter of $G[C]$ is at least 3. Then  $C$ can be partitioned into coalitions of diameter at most 2 without decreasing utilitarian welfare.
\end{lemma}

\begin{proof}

We first consider the case that $G[C]$ has a bridge $e=\{u,v\}$ and it connects two connected components of size at least 2. Let $C_u$ and $C_v$ be the connected components of $G[C]-e$ that contains $u$ and $v$, respectively. Then we can partition $C$ into two coalitions $C_u$ and $C_v$ without decreasing the utilitarian welfare.
Let $y_u$ and $z_u$ be the number of vertices and the number of edges in $G[C_u]$ and let $y_v, z_v$ be the number of vertices and  the number of edges in $G[C_v]$. 
Then the utilitarian welfare of $C$ is $\uw(C)=2(z_u+z_v+1)/(y_u+y_v)$.
On the other hand, the utilitarian welfare of $C_u$ is $\uw(C_u)=2z_u/y_u$  and the one of $C_v$ is $\uw(C_v)=2z_v/y_v$. Then, we have $(\uw(C_u)+\uw(C_v))-\uw(C)= {2z_u}/{y_u}+{2z_v}/{y_v}-{2(z_u+z_v+1)}/{(y_u+y_v)}
    = {(2y_u(y_u z_v-y_v)+2y_v^2 z_u)}/{y_u y_v(y_u+y_v)}\ge 0$.
The last inequality follows from $y_u,y_v\ge 2$ and  $z_v\ge y_v-1$. 
Thus, $C$ can be partitioned into two coalitions without decreasing the utilitarian welfare. 

Next, we consider the case that there is no bridge that connects two connected components of size at least 2 in $G[C]$.
Since $G$ is a block graph, $G[C]$ is also a block graph. Let $C_x$ be the largest clique of size $x(\ge 2)$ with only one cut vertex $v$ in $G[C]$. Such $v$ always exists unless $G[C]$ is not a clique. In other words, $C_x$ is a maximum-size leaf on the block-cut tree\footnote{The definition of a block-cut tree appears in Section~\ref{subsec:block-cut}} of $G[C]$. 
Let $C^{\prime}=C_x\setminus\{v\}$ and $C^{\prime\prime}=C\setminus C_x\cup\{v\}$. Then we show that $\uw(C^{\prime})+\uw(C^{\prime\prime})\ge\uw(C)$.

Let $y$ be the number of vertices and $z$ be the number of edges in $G[C^{\prime\prime}]$. 
Figure~\ref{fig:nobridge} shows the value of $x,y$ and $z$ with respect to the examples of $C^{\prime}$ and $C^{\prime\prime}$.
\begin{figure}[htb]
\centering
\includegraphics[width=0.5\textwidth]{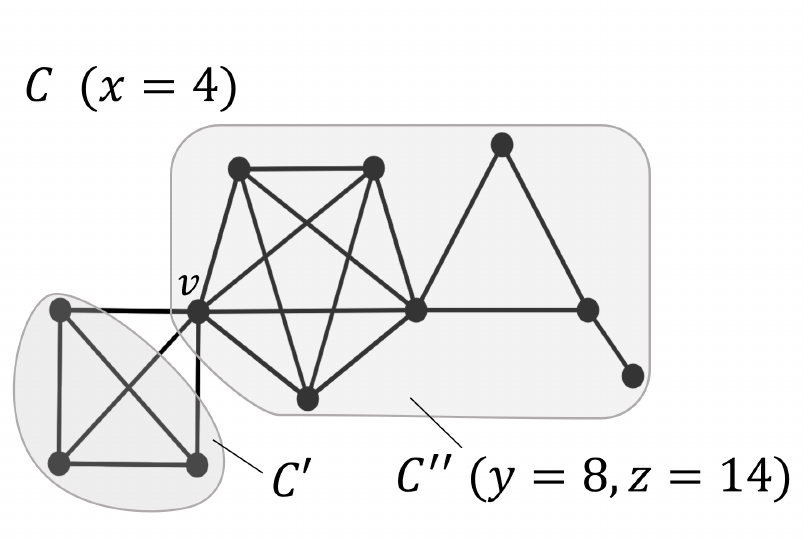}
\caption{The value of $x,y$ and $z$ with respect to the examples of $C^{\prime}$ and $C^{\prime\prime}$.} \label{fig:nobridge}
\end{figure}
By the same argument as in Lemma~\ref{lem:opt:diam2}, we obtain Equation~\eqref{align:opt:diam2}. Since $x\ge 2$ and $y\ge 2$, the sign of $(y^2(x-2)+2(z-y)(x-1))/y(x+y-1)$ depends on the sign of $z-y$.
Since the diameter of $G[C]$ is at least 3, $C_x$ forms a clique, and there is no bridge that connects two connected components of size at least 2 in $G[C]$, $G[C^{\prime\prime}]$ contains a cycle. Thus, we have $z-y\ge 0$. Therefore, we obtain $\uw(C^{\prime})+\uw(C^{\prime\prime})\ge\uw(C)$.


By performing these operations unless the diameter of  $G[C]$ is at most 2, we can partition $C$ into coalitions whose diameter is 2 without decreasing utilitarian welfare. 
\end{proof}

By combining Lemmas~\ref{lem:opt:diam2} and \ref{lem:opt:diam3}, we can obtain a coalition structure whose coalitions form stars or cliques from any coalition structure. This implies Theorem~\ref{thm:block:opt}.

\section{Maximizing utilitarian welfare on block graphs: Algorithm}\label{chap:block:algorithm}
In this section, we show \textsc{Utilitarian Welfare Maximization} on unweighted block graphs can be computed in polynomial time. We give a polynomial-time algorithm based on dynamic programming using optimal coalition structures shown in Theorem~\ref{thm:block:opt}. 

\begin{theorem}\label{thm:block:sw:opt}
\textsc{Utilitarian Welfare Maximization} on unweighted block graphs can be computed in time $O(n\Delta^4)$.
\end{theorem}

To design a DP-based algorithm, we first define the block-cut tree of a graph.


\subsection{Block-cut tree}\label{subsec:block-cut}
For a block graph $G=(V,E)$, let $\mathcal{B}=\{B_1,\ldots,B_\beta\}$ be the set of maximal cliques, called \emph{blocks}, in $G$ and $\cC=\{c_1,\ldots,c_{\gamma}\}$ be the set of cut vertices in $G$. Then the \emph{block-cut tree} $\mathcal{T}(G) = (\mathcal{X}, \mathcal{E})$ of $G$ is a tree such that $\mathcal{X}=\mathcal{B}\cup\mathcal{C}$ and each edge in $\mathcal{E}$ connects a block $B\in \mathcal{B}$ and a cut vertex $c\in \mathcal{C}\cap B$.
For simplicity, we sometimes write $\mathcal{T}$ instead of $\mathcal{T}(G)$. 
We call $B\in \mathcal{B}$ a block node and $c\in \mathcal{C}$ a cut node of $\mathcal{T}$. 
For convenience, we consider a block-cut tree rooted by a block node $B_r$. We denote by $\mathcal{T}_{x}(G)$ a subtree consisting of a node $x\in \cX$ and its descendants on $\cT$, by $p_{\cT}(x)$ the parent node of $x\in \cX\setminus \{B_r\}$, and by $\textrm{Child}_{\cT}(x)$ the set of children nodes of $x\in \cX$ in $\cT$ (see Figures~\ref{fig:tree:decomposition1} and \ref{fig:tree:decomposition2}). We also denote by $V_x$ the set of vertices in $G$ corresponding to $\cT_x$ and by $G[V_x]$ the induced subgraph of $G$ corresponding to $\cT_x$. For each block node $B\in \mathcal{B}\setminus \{B_r\}$ and its parent node $p_{\cT}(B)\in \mathcal{C}$, $c_p$ denotes the cut vertex in $B\cap p_{\cT}(B)$.
For $B\in \mathcal{B}$, we define $R(B):=B\setminus\cC$ as the set of non-cut vertices in $B$.
For a rooted block-cut tree, its leaves are block nodes.
\begin{figure}[htb]
\begin{tabular}{cc}
\begin{minipage}{0.45\hsize}
\centering
\includegraphics[width=\textwidth]{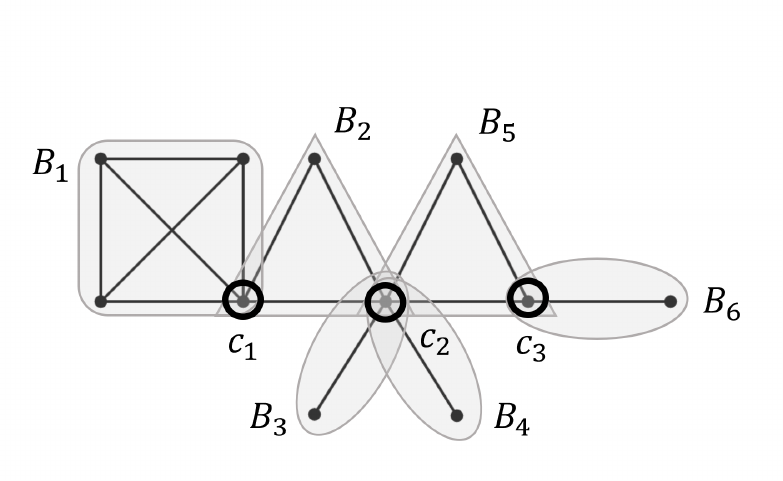}
\caption{Example of a block graph.} \label{fig:tree:decomposition1}
\end{minipage}
\begin{minipage}{0.45\hsize}
\centering
\includegraphics[width=\textwidth]{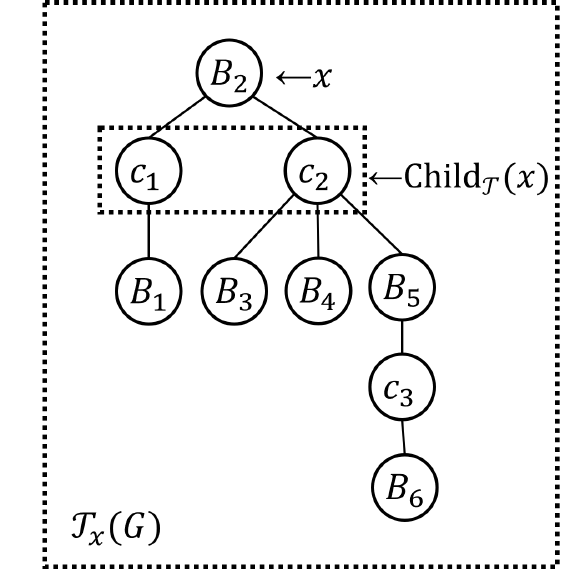}
\caption{The corresponding tree structure of block graph in Figure \ref{fig:tree:decomposition1}.}\label{fig:tree:decomposition2}
\end{minipage}
\end{tabular}
\end{figure}

\subsection{Recurrence relations: overviews}\label{subsec:DP:overviews}
By Theorem~\ref{thm:block:opt}, we design a dynamic programming algorithm on block-cut trees to find an optimal coalition structure such that each coalition is either a clique or a star. 

For a block node $B\in \mathcal{B}$, we define $\texttt{T}_B=\{\iso,\cl,\stc,\stl\}$ as the set of states of the cut vertex $c_p\in B\cap p_{\cT}(B)$ that represents the role of $c_p$ within the coalition containing $c_p$ in the intermediate steps of the algorithm for a block node. For a coalition structure $\PP_{B}$ of $G[V_B]$, $\iso$ means that $c_p$ is a singleton in $\PP_B$, $\cl$  means that $c_p$ belongs to a coalition that forms a clique of size at least 2 in $\PP_B$, $\stc$ means that  $c_p$ belongs to a coalition that forms a star (not a singleton) and $c_p$ is its center in $\PP_B$, and $\stl$ means that $c_p$ belongs to a coalition that forms a star (not a singleton) and $c_p$ is its leaf in $\PP_B$. Figure~\ref{fig:T_B} shows the role of $c_p\in B\cap p_{\cT}(B)$ with respect to $\texttt{T}_B$.
\begin{figure}[htb]
\centering
\includegraphics[width=\textwidth]{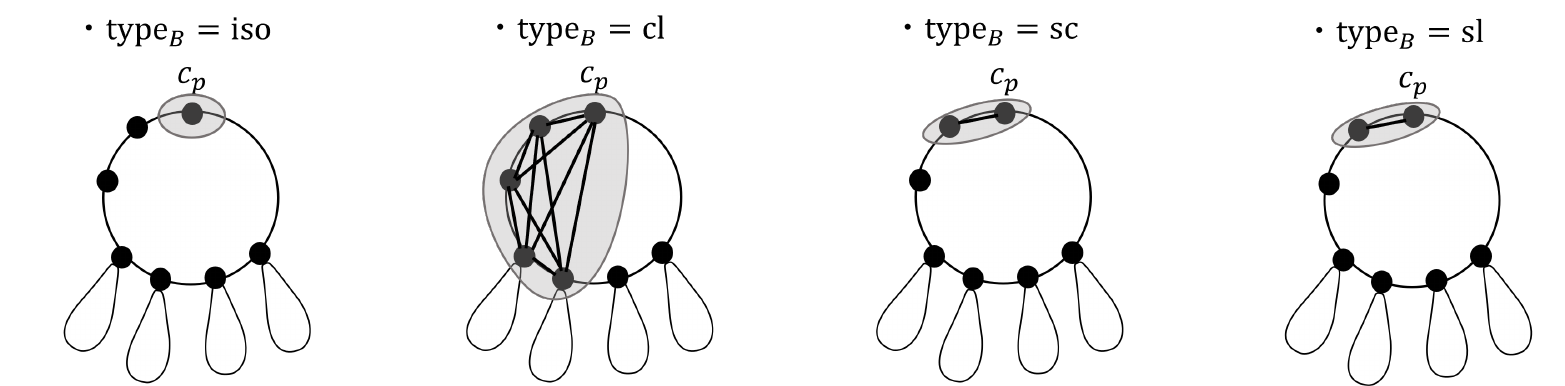}
\caption{The role of $c_p\in B\cap p_{\cT}(B)$ in $B$ with respect to $\texttt{T}_B$.} \label{fig:T_B}
\end{figure}

We also define $\texttt{T}_c=\{\iso,\cl,\stc_1,\ldots, \stc_{|\textrm{Child}_{\cT}(c)|},\stcu, \stl\}$ for a cut node $c\in \mathcal{C}$ as  the set of states of cut vertex $c$  that represents the role of  $c$ in the intermediate step of the algorithm.
Intuitively, in a cut node, the results of its children's block nodes  are integrated.
The states $\iso$, $\cl$, $\stl$ of cut vertex $c$ are the same meaning as the ones in $\texttt{T}_B$. The states $\stc_\ell$ denote that the cut vertex $c$ is the center vertex in the coalition that forms a star with $\ell$ leaves in $G[V_c]$. The state $\stcu$ denotes that the cut vertex $c$ is the center vertex in the coalition that forms a star with leaves in $G[V_c]$ and the coalition to which $c\in \mathcal{C}$ belongs \emph{will} contain one vertex (denote it $v_{\textrm{leaf}}$) in the parent node $p_{\mathcal{T}}(c)\in \mathcal{B}$ as a leaf.
Note that since the coalition forms a star and $p_{\mathcal{T}}(c)$ forms a clique, it can contain at most one vertex in $p_{\mathcal{T}}(c)$. 
For state $\stcu$,  we need not preserve the number of leaves of the star of $c$ because the maximum utilitarian welfare with respect to $\stcu$ can be computed from the values with respect to $\stc_\ell$ for $1\le \ell \le |\textrm{Child}_{\cT}(c)|$.
Figure~\ref{fig:T_c} shows the role of $c\in \mathcal{C}$ with respect to types in $\texttt{T}_c$.
\begin{figure}[htb]
\centering
\includegraphics[width=\textwidth]{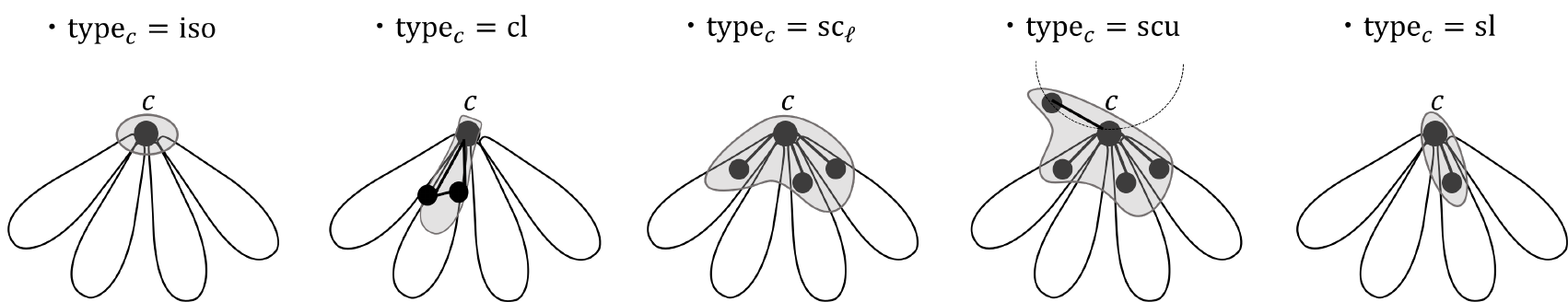}
\caption{The role of $c\in \mathcal{C}$ in its cut node with respect to $\texttt{T}_c$.} \label{fig:T_c}
\end{figure}

Now, we define the recurrence relations for block nodes and cut nodes. 
For block nodes, we define $f^*(B,\typeB,k)$ as the maximum utilitarian welfare among coalition structures in $G[V_{B}]$ that satisfies the following condition: 
\begin{itemize}
    \item $c_p$ is the parent cut node of $B$ and the role of the $c_p$ within the coalition to which $c_p$ belongs is $\typeB\in \{\iso, \cl, \stc,\stl\}$.
    \item if $\textrm{Child}_{\cT}(B)\neq\emptyset$, there are exactly $k$ cut vertices $c_1,\ldots c_k\in \textrm{Child}_{\cT}(B)$ such that the coalition that $c_i$ belongs to contain no vertex in $V_{c_i}\setminus \{c_{i}\}$ for $1\le i\le k$.
\end{itemize}


For a cut node $c\in C$ and $\typec\in\{\iso,\cl,\stc_1,\ldots,\stc_{|\textrm{Child}_{\cT}(c)|},\stl\}$, let $g^*(c,\typec)$ be the maximum utilitarian welfare in the coalition structures in $G[V_c]$ in which the role of $c$ is $\typec$.

Also, $g^*(c,\stcu)$ denotes the maximum utilitarian welfare in  $G[V_{c}\cup\{v_{\textrm{leaf}}\}]$ when the cut node $c\in C$ becomes the center of a star in its coalition and $v_{\textrm{leaf}}\in p_{\cT}(c)$ becomes a leaf of $c$'s star coalition.
Then we can compute $g^*(c,\stcu)$ from $g^*(c,\stc_\ell)$ for $1\le \ell \le |\textrm{Child}_{\cT}(c)|$ as follows:
\begin{align*}
g^*(c,\stcu)=\max_\ell\left\{g^*(c,\stc_\ell)+\frac{2}{(\ell+1)(\ell+2)}\right\}.
\end{align*}
Since the utilitarian welfare of a star with $\ell$ leaves and a star with $\ell+1$ leaves are $2\ell/(\ell+1)$ and $2(\ell+1)/(\ell+2)$, respectively, the increase of the utilitarian welfare when a leaf is added to the coalition that forms a star with $\ell$ leaves is $2(\ell+1)/(\ell+2)-2\ell/(\ell+1)=2/(\ell+1)(\ell+2)$. Thus, the equation holds.  

Let $g^*(c,\fin)=\max_{\typec\in\{\cl,\stc_1,\ldots,\stc_|\textrm{Child}_{\cT}(c)|,\stl\}}\{g^*(c,\typec)\}$ be the maximum utilitarian welfare in $G[V_c]$. If there is no coalition structure that satisfies the conditions, we set $f^*(B,\typeB,k)=-\infty$ and $g^*(c,\typec)=-\infty$ as invalid cases. 
In the following, we define the recurrence relations of the dynamic programming algorithm to compute $f^*(B,\typeB,k)$ and $g^*(c,\typec)$ recursively.

\subsection{Recurrence relations in the block nodes}
First, we prove that for such an optimal coalition structure, the following lemma holds. 
In $\mathcal{T}$, let $B\in \mathcal{B}$ be the parent node of the cut node $c\in \mathcal{C}$. Then, the following lemma holds.

\begin{lemma}\label{scu:num}
Let $B\in \mathcal{B}$ be a block node in a rooted block-cut tree.
Then there exists an optimal coalition structure such that at most one cut vertex $c$ in $\textrm{Child}_{\cT}(B)$ that is the center of a star having both vertices in $V_B\setminus B$ and in $B$ as leaves. 
\end{lemma}

\begin{proof}
For the sake of contradiction, suppose that there are $t(\ge 2)$ vertices $c_1,\ldots, c_t\in \textrm{Child}_{\cT}(B)$ that are the centers of $t$ stars having both vertices in $V_{\textrm{Child}_{\cT}(B)}$ and in $B$ as leaves in an optimal coalition structure $\PP^*$. Let $S_1,\ldots,S_t$ be $t$ stars whose centers are $c_1,\ldots, c_t$ 
and let $\ell_1+1,\ldots,\ell_t+1$ be the number of leaves of $S_1,\ldots,S_t$. 
Since $S_i$ is a star and $B$ forms a clique, each $S_i$ contains exactly one vertex in $B$ as a leaf.
We denote by $u_i$ the unique leaf in $B$ of $S_i$.

Now we reconstruct a coalition structure $\PP$ from $\PP^*$. 
Let $S'_i=S_i\setminus \{u_i\}$ for $1\le i\le t$ and $C_B = \{u_1,\ldots, u_t\}$.
Then $S'_i$ is a star with $\ell_i$ leaves and $C_B$ is a clique of size $t$.
Let $\PP = \PP^*\setminus \{S_1,\ldots,S_t\} \cup  \{S'_1,\ldots,S'_t,C_B\}$.

Considering $\ell_1,\ldots,\ell_t\ge 1$, the difference between $\PP$ and $\PP^*$ is as follows:
\begin{align*}
\uw(\PP^*) -\uw(\PP)
 &=\uw(S_1)+\cdots+\uw(S_t)-(\uw(S'_1)+\cdots+\uw(S'_t)+t-1)\\
 &=\frac{2}{(\ell_1+1)(\ell_1+2)}+\cdots+\frac{2}{(\ell_t+1)(\ell_t+2)}-t+1\\
&\le  -\frac{2}{3}t+1<0.
\end{align*}
This contradicts the optimality of $\PP^*$. 
\end{proof}

From Lemma~\ref{scu:num}, there is at most one cut vertex in $\textrm{Child}_{\cT}(B)$ whose role is $\stcu$ in the coalition structure that realizes $f^*(B,\typeB,k)$. Therefore, we consider two cases (1) no cut vertex in $\textrm{Child}_{\cT}(B)$ takes the role $\stcu$ and (2) exactly one cut vertex in $\textrm{Child}_{\cT}(B)$ takes the role $\stcu$. 

\subsubsection{Case (1): No cut vertex in $\textrm{Child}_{\cT}(B)$ takes the role $\stcu$.}\

First, we consider Case (1). 
For a block node $B$ and $c_1, c_2, \ldots, c_i \in \textrm{Child}_{\cT}(B)$ where $1\le i\le |\textrm{Child}_{\cT}(B)|$, we define  $f(B,\typeB,i,k)$ as the maximum utilitarian welfare  in $G[\{c_p\}\cup R(B)\cup \bigcup_{j=1}^{i} V_{c_j}]$ that satisfies the following conditions: 
\begin{itemize}
 \item $c_p$ is the parent cut node of $B$ and the role of the $c_p$ within the coalition to which $c_p$ belongs is $\typeB\in \{\iso, \cl, \stc,\stl\}$.
 \item There are exactly $k$ cut vertices among $c_1, c_2, \ldots, c_i$ such that each of $k$ such cut vertices does not form the coalition with vertices in its subtree, and each of exactly $i-k$ other cut vertices among $c_1, c_2, \ldots, c_i$ forms a coalition with only vertices in its subtree. 
\end{itemize}
Note that $0\le k\le i \le |\textrm{Child}_{\cT}(B)|$. Then, the maximum utilitarian welfare in $G[V_B]$ in Case (1) can be expressed by $f(B,\typeB,|\textrm{Child}_{\cT}(B)|,k)$. 
We can compute $f(B,\typeB,|\textrm{Child}_{\cT}(B)|,k)$ by the following argument.

As the base case, we consider the maximum utilitarian welfare in  $G[\{c_p\}\cup R(B)]$. For convenience, we set this case as $i=0$. Considering $G[\{c_p\}\cup R(B)]$ is a clique, we have 
\begin{align*}
 f(B,\typeB,0,0) &= \begin{cases}
0 & \text{if $R(B)=\emptyset$}\\ 
|R(B)|-1 & \text{if $R(B)\neq\emptyset$ \& $\typeB= \iso$}\\ 
|R(B)| & \text{if $R(B)\neq\emptyset$ \& $\typeB= \cl$}\\
\max\{1,|R(B)|-1\} & \text{otherwise}. 
\end{cases}
\end{align*}
The above equation holds by the following observation: 
In the first case, $G[\{c_p\}\cup R(B)]$ consists of only one vertex, and hence $f=0$. 
If $R(B)\neq \emptyset$,  $R(B)$ is contained in some coalition in any optimal coalition structure in $G[\{c_p\}\cup R(B)]$.
In the second case, $c_p$ is a singleton. Thus, the maximum utilitarian welfare, in this case, is  $|R(B)|-1$. Note that $R(B)$ forms a clique. In the third case, the optimal coalition structure with the maximum utilitarian welfare is clearly $\{c_p\}\cup R(B)$ and it forms a clique. Thus, the maximum utilitarian welfare is $|R(B)|$. 
In the remaining case, $c_p$ takes the state $\stc$ or $\stl$, which means that $c_p$'s coalition forms a star. In this case, if $|R(B)|=1$,  an optimal coalition structure forms a star of size 2. Thus, the utilitarian welfare is $1$. 
On the other hand, if $|R(B)|>1$, an optimal coalition structure forms a star of size 2 which contains $c_p$ and a clique of size $|R(B)|-1$. Thus, the utilitarian welfare is  $|R(B)|-1$.

Next, we consider the case when $i> 0$. 
\paragraph{Case $\typeB=\iso$:}\ 

We define the recurrence relation for $\typeB=\iso$ as follows.
\begin{align*}
&f(B,\iso,i,k)\\
&=
\begin{cases}
f(B,\iso,i-1,0)+g^*(c_i,\fin)
& \text{if $k = 0$}\\
\max\left\{
\begin{array}{c}
f(B,\iso,i-1,0)+g^*(c_i,\iso),\\
f(B,\iso,i-1,1)+g^*(c_i,\fin)
\end{array}
\right\}
& \text{if $k=1$ \& $R(B)= \emptyset$}\\
\max\left\{
\begin{array}{c}
f(B,\iso,i-1,k-1)+g^*(c_i,\iso)+1,\\
f(B,\iso,i-1,k)+g^*(c_i,\fin)
\end{array}
\right\}
& \text{otherwise}.
\end{cases}
\end{align*}
The above equation holds by the following observation: 
Consider the difference between the coalition structure in $G[\{c_p\}\cup R(B)\bigcup_{j=1}^{i} V_{c_j}]$ and 
the one in $G[\{c_p\}\cup R(B)\bigcup_{j=1}^{i-1} V_{c_j}]$. 
First, when $k=0$, the cut node $c_i$ in $G[V_{c_i}]$ is not in the same coalition with any vertex in $B$. 

When $k\ge 1$, the cut node $c_i$ in $G[V_{c_i}]$ belongs to either a coalition containing a vertex in $B$ ($g^*(c_i,\iso)$) or not ($g^*(c_i,\fin)$). Thus, we take the maximum one.
Suppose that $k=1$ and $R(B)=\emptyset$. We have two choices: $c_i$ is a singleton or $c_i$ belongs to a coalition containing only vertices in $V_{c_i}$.
On the other hand, consider the case when $k\ge 2$ or $R(B)\neq \emptyset$. Then if $c_i$ belongs to a coalition with vertices in $B$, its coalition forms a clique of size at least $2$. Thus the utilitarian welfare will increase by $1$. 

\paragraph{Case $\typeB=\cl$}\

We define the recurrence relation for $\typeB=\cl$ as follows.
\begin{align*}
f(B,\cl,i,k)=&
\begin{cases}
f(B,\cl,i-1,0)+g^*(c_i,\fin)
& \text{if $k = 0$}\\
\max\left\{
\begin{array}{c}
f(B,\cl,i-1,k-1)+g^*(c_i,\iso)+1,\\
f(B,\cl,i-1,k)+g^*(c_i,\fin)
\end{array}
\right\}
& \text{otherwise}.
\end{cases}
\end{align*}
In a valid case, we define $f(B,\cl,|\textrm{Child}_{\cT}(B)|,0)=-\infty$ if  $R(B)=\emptyset$ since $c_p$ cannot be a clique of size at least 2. 

The above equation holds by the following observation: 
We consider the difference between the coalition structures in $G[\{c_p\}\cup R(B)\cup \bigcup_{j=1}^{i} V_{c_j}]$ and $G[\{c_p\}\cup R(B)\cup \bigcup_{j=1}^{i-1} V_{c_j}]$. 
When $k=0$, the cut node $c_i$ in $G[V_{c_i}]$ is not in the same coalition with any vertex in $B$.

When $k\ge 1$, the cut node $c_i$ in $G[V_{c_i}]$ belongs to either a coalition containing a vertex in $B$ ($g^*(c_i,\iso)$) or not ($g^*(c_i,\fin)$). In the former case, $c_i$ joins a coalition that forms a clique in $B$. Thus the utilitarian welfare will increase by $1$. 
The latter case is that $c_i$ belongs to a coalition containing only vertices in $V_{c_i}$.
Thus, the utilitarian welfare can be computed by adding the coalition structure in $G[V_{c_i}]$.


\paragraph{Case $\typeB=\stc$}\

We define the recurrence relation for $\typeB=\stc$ as follows.



\begin{flalign}
&f(B,\stc,i,k)&\nonumber
\end{flalign}
\begin{empheq}
[left = {
= \empheqlbrace \,}]{alignat=2}
& f(B,\stc,i-1,0)+g^*(c_i,\fin)  &\quad &    \text{if} \ k=0 \nonumber\\ 
& \max\left\{
\begin{array}{c}
f(B,\stc,i-1,k-1)+g^*(c_i,\iso),\\
f(B,\stc,i-1,k)+g^*(c_i,\fin)
\end{array}
\right\}    &     &    
\left\{\begin{aligned}
        & \text{if}\  k=1 \& |R(B)|=1,  \\
        &  \text{if}\  k=2 \& R(B)=\emptyset  \nonumber
\end{aligned} \right.\\
& \max\left\{
\begin{array}{c}
f(B,\stc,i-1,k-1)+g^*(c_i,\iso)+1,\\
f(B,\stc,i-1,k)+g^*(c_i,\fin)
\end{array}
\right\}&\quad &    \text{otherwise.} \nonumber
\end{empheq}



As an invalid case, we define $f(B,\stc,|\textrm{Child}_{\cT}(B)|,0)=-\infty$ if $R(B)=\emptyset$, since $c_p$ cannot be the center of a star.

The above equation holds by the following observation: 
We consider the difference between the coalition structures in $G[\{c_p\}\cup R(B)\cup \bigcup_{j=1}^{i} V_{c_j}]$ and $G[\{c_p\}\cup R(B)\cup \bigcup_{j=1}^{i-1}V_{c_j}]$. 
When $k=0$, the cut node $c_i$ in $G[V_{c_i}]$ is not in the same coalition with any vertex in $B$.

When $k\ge 1$, $c_i$ is either in the same coalition with vertices in $B$ ($g^*(c_i,\iso)$) or not ($g^*(c_i,\fin)$). 
Consider the case where $c_i$ joins the same coalition with vertices in $B$ when $k+|R(B)|\neq2$. 
In this case, $c_p$ forms a clique of size $2$ with some cut vertex in $B$ or a vertex in $R(B)$ and other vertices in $B$ form a clique in an optimal coalition structure. Thus, the utilitarian welfare increases by $g^*(c_i,\iso)+1$ if $c_i$ joins a coalition in $B$.

If $k=1$ and $|R(B)|=1$, $c_p$ forms a star of size 2 with the vertex in $R(B)$ and the cut vertex in $B$ forms a singleton in an optimal coalition structure.
By the definition of $f$ in the base case, the utilitarian welfare of the star is already accounted for. 
Moreover, if $k=2$ and $R(B)=\emptyset$, $c_p$ forms a star of size 2 with a cut vertex in $B$ and the other cut vertex in $B$ forms a singleton. 



\paragraph{Case $\typeB=\stl$}\

We define the recurrence relation for $\typeB=\stl$ as follows.


\begin{flalign}
&f(B,\stl,i,k)&\nonumber
\end{flalign}
\begin{empheq}
[left = {
= \empheqlbrace \,}]{alignat=2}
& f(B,\stl,i-1,0)+g^*(c_i,\fin)  &\quad &    \text{if} \ k=0 \nonumber\\ 
& \max\left\{
\begin{array}{c}
f(B,\stl,i-1,k-1)+g^*(c_i,\iso),\\
f(B,\stl,i-1,k)+g^*(c_i,\fin)
\end{array}
\right\}    &     &    
\left\{\begin{aligned}
        & \text{if}\  k=1 \& |R(B)|=1,  \\
        &  \text{if}\  k=2 \& R(B)=\emptyset  \nonumber
\end{aligned} \right.\\
& \max\left\{
\begin{array}{c}
f(B,\stl,i-1,k-1)+g^*(c_i,\iso)+1,\\
f(B,\stl,i-1,k)+g^*(c_i,\fin)
\end{array}
\right\}&\quad &    \text{otherwise.} \nonumber
\end{empheq}

As an invalid case, we deifne $f(B,\stl,|\textrm{Child}_{\cT}(B)|,0)=-\infty$ if $R(B)=\emptyset$, since $c_p$ cannot be a leaf of a star. 

The above equation holds by the following observation: 
When $k=0$, the cut node $c_i$ in $G[V_{c_i}]$ is not in the same coalition with any vertex in $B$.

When $k\ge 1$, in  a coalition structure realizing $f(B,\stl,i,k)$, $c_p$ forms a star of size 2 with a vertex in $\textrm{Child}_{\cT}(B)$ whose role is $\iso$ or a vertex in $R(B)$ that becomes the center of the star. 
By similar reasons to the case of $\typeB=\stc$, the correctness of the recurrence relations follows.




\subsubsection{Case (2): Exactly one cut vertex in $\textrm{Child}_{\cT}(B)$ takes the role $\stcu$.}\

Next, we consider Case (2). Fix a cut vertex $c_{\textrm{center}}$ in $\textrm{Child}_{\cT}(B)$ whose role is $\stcu$, that is, the center of a star having the unique leaf $v_{\textrm{leaf}}$ in $B$. 
The number of  choices of such pairs is at most $\Delta^2$.

Then  we define $f_{c_{\textrm{center}}}^{v_{\textrm{leaf}}}(B,\typeB,i,k)$ as the maximum utilitarian welfare  in $G[\{c_p\}\cup R(B)\cup \bigcup_{j=1}^{i} V_{c_j}]$ that satisfies the following conditions:
\begin{itemize}
    \item $c_p$ is the parent cut node of $B$ and the role of $c_p$ is $\typeB\in \{\iso, \cl, \stc,\stl\}$.
    \item The role of $c_{\textrm{center}}$ is $\stcu$ and $v_{\textrm{leaf}}$ belongs to $c_{\textrm{center}}$'s coalition  as a leaf. 

\item The subgraph has $i$ cut vertices $c_1,c_2,\ldots,c_i$ in $\textrm{Child}_{\cT}(B)$. One is $c_{\textrm{center}}$. The coalitions that the $i-1$ remaining cut vertices belong to form as follows: there exist exactly $k$ cut vertices, which may include $v_{\textrm{leaf}}$ if $v_{\textrm{leaf}}\in \textrm{Child}_{\cT}(B)$, such that each of them does not form the coalition with vertices in its subtree. In other words, the other $i-k-1$ cut vertices form their coalitions with vertices in their subtrees.
\end{itemize}

Without loss of generality, we suppose $c_1 = c_{\textrm{center}}$ and $c_2 = v_{\textrm{leaf}}$ if $v_{\textrm{leaf}}\in \textrm{Child}_{\cT}(B)$.
As the base case, we consider the maximum utilitarian welfare in $G[\{c_p\}\cup R(B)\cup V_{c_1}]$ if $v_{\textrm{leaf}}\in R(B)\cup\{c_p\}$   and $G[\{c_p\}\cup R(B)\cup \bigcup_{j=1}^{2} V_{c_j}]$ if $v_{\textrm{leaf}}\in \textrm{Child}_{\cT}(B)\setminus\{c_1\}$. 
For convenience, we set these cases as $i=1,2$ respectively. 
We first consider the case when $v_{\textrm{leaf}}\in R(B)\cup\{c_p\}$. 
By the assumption of $c_1 = c_{\textrm{center}}$, we set $f_{c_{\textrm{center}}}^{v_{\textrm{leaf}}}(B,\typeB,1,1)=-\infty$. Also the case of $v_{\textrm{leaf}}=c_p$ and $\typeB\in\{\iso,\cl,\stc\}$ contradicts the definition, so we set $f_{c_{\textrm{center}}}^{v_{\textrm{leaf}}}(B,\typeB,1,1)=-\infty$ for this case. In the remaining cases, we have  

\begin{flalign}
&f_{c_{\textrm{center}}}^{v_{\textrm{leaf}}}(B,\typeB,1,0)&\nonumber
\end{flalign}
\begin{empheq}
[left = {
= \empheqlbrace \,}]{alignat=2}
& g^*(c_1,\stcu)     &     &    
\left\{\begin{aligned}
        & \text{if}\  v_{\textrm{leaf}}\in R(B)  \&\  |R(B)|=1 \\
        & \ \&\  \typeB\in\{\iso,\stc,\stl\},   \\
        & \text{if}\  v_{\textrm{leaf}}=c_p \ \& \ R(B)=\emptyset \\
        &  \ \&\  \typeB=\stl \nonumber
\end{aligned} \right.\\
& g^*(c_1,\stcu)+|R(B)|-2     &     &    
\left\{\begin{aligned}
        & \text{if}\  v_{\textrm{leaf}}\in R(B)\  \& \ |R(B)|>1  \\
        & \ \&\  \typeB= \iso  \nonumber
\end{aligned} \right.\\
& g^*(c_1,\stcu)+|R(B)|-1     &     &    
\left\{\begin{aligned}
        & \text{if}\  v_{\textrm{leaf}}\in R(B)\  \& \ \typeB= \cl, \\
        & \text{if}\  v_{\textrm{leaf}}=c_p \  \&\  R(B)\neq\emptyset \\
        & \ \&\  \typeB=\stl  \nonumber
\end{aligned} \right.\\
& g^*(c_1,\stcu)+\max\{1,|R(B)|-2\}    &\quad &    \text{otherwise.} \nonumber
\end{empheq}
If $v_{\textrm{leaf}}\in \textrm{Child}_{\cT}(B)\setminus\{c_1\}$ (i.e., $i=2$), $k$ can be $0,1,2$, but actually $k$ should be $1$ by definition. 
We thus set $f_{c_{\textrm{center}}}^{v_{\textrm{leaf}}}(B,\typeB,2,k)=-\infty$ for $k=0,2$. For the remaining case, 
\begin{flalign}
&f_{c_{\textrm{center}}}^{v_{\textrm{leaf}}}(B,\typeB,2,1)&\nonumber
\end{flalign}
\begin{empheq}
[left = {
= \empheqlbrace \,}]{alignat=2}
& g^*(c_1,\stcu) +g^*(c_2,\iso)+\max\{0,|R(B)|-1\}    &     & \text{if}\  \typeB=\iso \nonumber \\
& g^*(c_1,\stcu) +g^*(c_2,\iso)+|R(B)|     &     & \text{if}\    \typeB=\cl \nonumber\\
& g^*(c_1,\stcu) +g^*(c_2,\iso)     &     & 
 \left\{\begin{aligned}
  & \text{if}\ \typeB\in\{\stc,\stl\} \\
  &\ \& \ R(B)=\emptyset \nonumber
\end{aligned} \right.\\
& g^*(c_1,\stcu) +g^*(c_2,\iso)+\max\{1,|R(B)|-1\}    &\quad &    \text{otherwise.} \nonumber
\end{empheq}

Function $f_{c_{\textrm{center}}}^{v_{\textrm{leaf}}}(B,\typeB,i,k)$ can be calculated recursively in the same way as $f(B,\typeB,i,k)$ by using $g^*$. 
Indeed, as with the computation of $f(B,\typeB,i,k)$, the recurrence relations computing $f_{c_{\textrm{center}}}^{v_{\textrm{leaf}}}(B,\typeB,i,k)$ can be defined with respect to $\typeB$ with slight modification. If $\typeB\in\{\iso,\cl,\stc\}$, we define $f_{c_{\textrm{center}}}^{v_{\textrm{leaf}}}(B,\typeB,i,k)=-\infty$ if $v_{\textrm{leaf}}=c_p$ as invalid cases.

First, we consider the case when $i\ge 2$ and $v_{\textrm{leaf}}\in R(B)\cup\{c_p\}$. 
\paragraph{Case $\typeB=\iso\  \&\  v_{\textrm{leaf}}\in R(B)$:}\ 

\begin{align*}
&f_{c_{\textrm{center}}}^{v_{\textrm{leaf}}}(B,\iso,i,k)\\
&=
\begin{cases}
f_{c_{\textrm{center}}}^{v_{\textrm{leaf}}}(B,\iso,i-1,0)+g^*(c_i,\fin)
& \text{if $k=0$}\\
\max\left\{
\begin{array}{c}
f_{c_{\textrm{center}}}^{v_{\textrm{leaf}}}(B,\iso,i-1,0)+g^*(c_i,\iso),\\
f_{c_{\textrm{center}}}^{v_{\textrm{leaf}}}(B,\iso,i-1,1)+g^*(c_i,\fin)
\end{array}
\right\}
& \text{if $k=1$ \& $|R(B)|=1$}\\
\max\left\{
\begin{array}{c}
f_{c_{\textrm{center}}}^{v_{\textrm{leaf}}}(B,\iso,i-1,k-1)+g^*(c_i,\iso)+1,\\
f_{c_{\textrm{center}}}^{v_{\textrm{leaf}}}(B,\iso,i-1,k)+g^*(c_i,\fin)
\end{array}
\right\}
& \text{otherwise}.
\end{cases}
\end{align*}

\paragraph{Case $\typeB=\cl\  \&\  v_{\textrm{leaf}}\in R(B)$:}\ 

\begin{align*}
&f_{c_{\textrm{center}}}^{v_{\textrm{leaf}}}(B,\cl,i,k)\\
&=
\begin{cases}
f_{c_{\textrm{center}}}^{v_{\textrm{leaf}}}(B,\cl,i-1,0)+g^*(c_i,\fin)
& \text{if $k=0$}\\
\max\left\{
\begin{array}{c}
f_{c_{\textrm{center}}}^{v_{\textrm{leaf}}}(B,\cl,i-1,k-1)+g^*(c_i,\iso)+1,\\
f_{c_{\textrm{center}}}^{v_{\textrm{leaf}}}(B,\cl,i-1,k)+g^*(c_i,\fin)
\end{array}
\right\}
& \text{otherwise}.
\end{cases}
\end{align*}

\paragraph{Case $\typeB=\stc\  \&\  v_{\textrm{leaf}}\in R(B)$:}\ 


\begin{flalign}
&f_{c_{\textrm{center}}}^{v_{\textrm{leaf}}}(B,\stc,i,k)&\nonumber
\end{flalign}
\begin{empheq}
[left = {
= \empheqlbrace \,}]{alignat=2}
& f_{c_{\textrm{center}}}^{v_{\textrm{leaf}}}(B,\stc,i-1,0)+g^*(c_i,\fin)    &\quad &    \text{if}\ k=0 \nonumber \\
& \max\left\{
\begin{array}{c}
f_{c_{\textrm{center}}}^{v_{\textrm{leaf}}}(B,\stc,i-1,k-1)+g^*(c_i,\iso),\\
f_{c_{\textrm{center}}}^{v_{\textrm{leaf}}}(B,\stc,i-1,k)+g^*(c_i,\fin)
\end{array}
\right\}     &     &    
\left\{\begin{aligned}
        & \text{if}\  k=1 \& |R(B)|=2,  \\
        &  \text{if}\  k=2 \& |R(B)|=1  \nonumber
\end{aligned} \right.\\
& \max\left\{
\begin{array}{c}
f_{c_{\textrm{center}}}^{v_{\textrm{leaf}}}(B,\stc,i-1,k-1)+g^*(c_i,\iso)+1,\\
f_{c_{\textrm{center}}}^{v_{\textrm{leaf}}}(B,\stc,i-1,k)+g^*(c_i,\fin)
\end{array}
\right\}   &\quad &    \text{otherwise.} \nonumber
\end{empheq}

\paragraph{Case $\typeB=\stl\  \&\  v_{\textrm{leaf}}\in R(B)$:}\ 


\begin{flalign}
&f_{c_{\textrm{center}}}^{v_{\textrm{leaf}}}(B,\stl,i,k)&\nonumber
\end{flalign}
\begin{empheq}
[left = {
= \empheqlbrace \,}]{alignat=2}
& f_{c_{\textrm{center}}}^{v_{\textrm{leaf}}}(B,\stl,i-1,0)+g^*(c_i,\fin)    &\quad &    \text{if}\ k=0 \nonumber \\
& \max\left\{
\begin{array}{c}
f_{c_{\textrm{center}}}^{v_{\textrm{leaf}}}(B,\stl,i-1,k-1)+g^*(c_i,\iso),\\
f_{c_{\textrm{center}}}^{v_{\textrm{leaf}}}(B,\stl,i-1,k)+g^*(c_i,\fin)
\end{array}
\right\}    &     &    
\left\{\begin{aligned}
        & \text{if}\  k=1 \& |R(B)|=2,  \\
        &  \text{if}\  k=2 \& |R(B)|=1  \nonumber
\end{aligned} \right.\\
& \max\left\{
\begin{array}{c}
f_{c_{\textrm{center}}}^{v_{\textrm{leaf}}}(B,\stl,i-1,k-1)+g^*(c_i,\iso)+1,\\
f_{c_{\textrm{center}}}^{v_{\textrm{leaf}}}(B,\stl,i-1,k)+g^*(c_i,\fin)
\end{array}
\right\}  &\quad &    \text{otherwise.} \nonumber
\end{empheq}

\paragraph{Case $\typeB=\stl\  \&\  v_{\textrm{leaf}}=c_p$:}\ 

\begin{align*}
&f_{c_{\textrm{center}}}^{v_{\textrm{leaf}}}(B,\stl,i,k)\\
&=
\begin{cases}
f_{c_{\textrm{center}}}^{v_{\textrm{leaf}}}(B,\stl,i-1,0)+g^*(c_i,\fin)
& \text{if $k=0$}\\
\max\left\{
\begin{array}{c}
f_{c_{\textrm{center}}}^{v_{\textrm{leaf}}}(B,\stl,i-1,0)+g^*(c_i,\iso),\\
f_{c_{\textrm{center}}}^{v_{\textrm{leaf}}}(B,\stl,i-1,1)+g^*(c_i,\fin)
\end{array}
\right\}
& \text{if $k=1$ \& $R(B)=\emptyset$}\\
\max\left\{
\begin{array}{c}
f_{c_{\textrm{center}}}^{v_{\textrm{leaf}}}(B,\stl,i-1,k-1)+g^*(c_i,\iso)+1,\\
f_{c_{\textrm{center}}}^{v_{\textrm{leaf}}}(B,\stl,i-1,k)+g^*(c_i,\fin)
\end{array}
\right\}
& \text{otherwise}.
\end{cases}
\end{align*}

Next, we consider the case when $i\ge 3$ and $v_{\textrm{leaf}}\in\textrm{Child}_{\cT}(B)\setminus\{c_1\}$. 
\paragraph{Case $\typeB=\iso$:}\ 

\begin{align*}
&f_{c_{\textrm{center}}}^{v_{\textrm{leaf}}}(B,\iso,i,k)\\
&=
\begin{cases}
-\infty & \text{if $k=0$}\\
f_{c_{\textrm{center}}}^{v_{\textrm{leaf}}}(B,\iso,i-1,1)+g^*(c_i,\fin)
& \text{if $k=1$}\\
\max\left\{
\begin{array}{c}
f_{c_{\textrm{center}}}^{v_{\textrm{leaf}}}(B,\iso,i-1,1)+g^*(c_i,\iso),\\
f_{c_{\textrm{center}}}^{v_{\textrm{leaf}}}(B,\iso,i-1,2)+g^*(c_i,\fin)
\end{array}
\right\}
& \text{if $k=2$ \& $R(B)=\emptyset$}\\
\max\left\{
\begin{array}{c}
f_{c_{\textrm{center}}}^{v_{\textrm{leaf}}}(B,\iso,i-1,k-1)+g^*(c_i,\iso)+1,\\
f_{c_{\textrm{center}}}^{v_{\textrm{leaf}}}(B,\iso,i-1,k)+g^*(c_i,\fin)
\end{array}
\right\}
& \text{otherwise}.
\end{cases}
\end{align*}

\paragraph{Case $\typeB=\cl$:}\ 

\begin{align*}
&f_{c_{\textrm{center}}}^{v_{\textrm{leaf}}}(B,\cl,i,k)\\
&=
\begin{cases}
-\infty & \text{if $k=0$}\\
f_{c_{\textrm{center}}}^{v_{\textrm{leaf}}}(B,\cl,i-1,1)+g^*(c_i,\fin)
& \text{if $k=1$}\\
\max\left\{
\begin{array}{c}
f_{c_{\textrm{center}}}^{v_{\textrm{leaf}}}(B,\cl,i-1,k-1)+g^*(c_i,\iso)+1,\\
f_{c_{\textrm{center}}}^{v_{\textrm{leaf}}}(B,\cl,i-1,k)+g^*(c_i,\fin)
\end{array}
\right\}
& \text{otherwise}.
\end{cases}
\end{align*}

\paragraph{Case $\typeB=\stc$:}\ 


\begin{flalign}
&f_{c_{\textrm{center}}}^{v_{\textrm{leaf}}}(B,\stc,i,k)&\nonumber
\end{flalign}
\begin{empheq}
[left = {
= \empheqlbrace \,}]{alignat=2}
& -\infty  &\quad &    \text{if} \ k=0 \nonumber\\ 
& f_{c_{\textrm{center}}}^{v_{\textrm{leaf}}}(B,\stc,i-1,1)+g^*(c_i,\fin)   &\quad &    \text{if}\ k=1 \nonumber \\
& \max\left\{
\begin{array}{c}
f_{c_{\textrm{center}}}^{v_{\textrm{leaf}}}(B,\stc,i-1,k-1)+g^*(c_i,\iso),\\
f_{c_{\textrm{center}}}^{v_{\textrm{leaf}}}(B,\stc,i-1,k)+g^*(c_i,\fin)
\end{array}
\right\}    &     &    
\left\{\begin{aligned}
        & \text{if}\  k=2 \& |R(B)|=1,  \\
        &  \text{if}\  k=3 \& R(B)=\emptyset  \nonumber
\end{aligned} \right.\\
& \max\left\{
\begin{array}{c}
f_{c_{\textrm{center}}}^{v_{\textrm{leaf}}}(B,\stc,i-1,k-1)+g^*(c_i,\iso)+1,\\
f_{c_{\textrm{center}}}^{v_{\textrm{leaf}}}(B,\stc,i-1,k)+g^*(c_i,\fin)
\end{array}
\right\} &\quad &    \text{otherwise.} \nonumber
\end{empheq}

\paragraph{Case $\typeB=\stl$:}\ 


\begin{flalign}
&f_{c_{\textrm{center}}}^{v_{\textrm{leaf}}}(B,\stl,i,k)&\nonumber
\end{flalign}
\begin{empheq}
[left = {
= \empheqlbrace \,}]{alignat=2}
& -\infty  &\quad &    \text{if} \ k=0 \nonumber\\ 
& f_{c_{\textrm{center}}}^{v_{\textrm{leaf}}}(B,\stl,i-1,1)+g^*(c_i,\fin)   &\quad &    \text{if}\ k=1 \nonumber \\
& \max\left\{
\begin{array}{c}
f_{c_{\textrm{center}}}^{v_{\textrm{leaf}}}(B,\stl,i-1,k-1)+g^*(c_i,\iso),\\
f_{c_{\textrm{center}}}^{v_{\textrm{leaf}}}(B,\stl,i-1,k)+g^*(c_i,\fin)
\end{array}
\right\}    &     &    
\left\{\begin{aligned}
        & \text{if}\  k=2 \& |R(B)|=1,  \\
        &  \text{if}\  k=3 \& R(B)=\emptyset  \nonumber
\end{aligned} \right.\\
& \max\left\{
\begin{array}{c}
f_{c_{\textrm{center}}}^{v_{\textrm{leaf}}}(B,\stl,i-1,k-1)+g^*(c_i,\iso)+1,\\
f_{c_{\textrm{center}}}^{v_{\textrm{leaf}}}(B,\stl,i-1,k)+g^*(c_i,\fin)
\end{array}
\right\} &\quad &    \text{otherwise.} \nonumber
\end{empheq}


Here, we define $f_{c_{\textrm{center}}}(B,\typeB,|\textrm{Child}_{\cT}(B)|,k)
=\max_{v_{\textrm{leaf}}\in B\setminus\{c_{\textrm{center}}\}}\break
\{f_{c_{\textrm{center}}}^{v_{\textrm{leaf}}}(B,\typeB,|\textrm{Child}_{\cT}(B)|,k)\}$.

\subsubsection{Computing $f^*$}

After the computation of $f$ (Case (1)) and $f_{c_{\textrm{center}}}^{v_{\textrm{leaf}}}$  (Case (2)) , we compute $f^*$ by taking the maximum of them as follows. 
\begin{align*}
&f^*(B,\typeB,k)\\
&=\max\{f(B,\typeB,|\textrm{Child}_{\cT}(B)|,k),f_{c_1}(B,\typeB,|\textrm{Child}_{\cT}(B)|,k),\\
&\ldots,f_{c_{|\textrm{Child}_{\cT}(B)|}}(B,\typeB,|\textrm{Child}_{\cT}(B)|,k)\}. 
\end{align*}

\subsection{Recurrence relation in the cut nodes}
In this subsection, we will see how we compute  $g^*(c,\typec)$. 
As seen before, $\typec$ can take a value in $\{\stcu, \iso,\cl,\stc_1,\ldots,\stc_{|\textrm{Child}_{\cT}(c)|},\stl\}$. 
Since we see that $g^*(c,\stcu)$ can be computed from $g^*(c,\stc_\ell)$ for $1\le \ell \le |\textrm{Child}_{\cT}(c)|$ in Section \ref{subsec:DP:overviews}, we consider the remaining cases.
To do this, we define an auxiliary function $g(c,\typec,i)$. 
Let $B_1, B_2, \ldots, B_{|\textrm{Child}_{\cT}(c)|}$ be the block nodes belonging to $\textrm{Child}_{\cT}(c)$. For $B_1, B_2, \ldots, B_i \in \textrm{Child}_{\cT}(c)$, we define $g(c,\typec,i)$ as the maximum utilitarian welfare of coalition structures in $G[\bigcup_{j=1}^{i}V_{B_j}]$ in which the role of $c$ is $\typec$. 
By the definition of $g(c,\typec,i)$, we have:
\[
g^*(c,\typec) = g(c,\typec,|\textrm{Child}_{\cT}(c)|).
\]
Then we recursively compute each of $g(c,\typec,|\textrm{Child}_{\cT}(c)|)$'s  as follows.

\paragraph{Case $\typec=\iso$}\

We define the recurrence relation for $\typec=\iso$ as follows.
\begin{align*}
g(c,\iso,i)=& \begin{cases}
0 & i = 0 \\
g(c,\iso,i-1)+\max_{k}\{f^*(B_i,\iso,k)\}
& i> 0.
\end{cases}
\end{align*}
%
In this case, $c$ is a singleton due to $\typec=\iso$. Under this constraint, the coalition structure with the maximum utilitarian welfare forms the singleton $c$ plus the collection of the best coalitions of $G[V_{B_1}], G[V_{B_2}], \ldots, G[V_{B_i}]$.  
When $i=0$, let $g(c,\iso,0)=0$ as the base case.  
When $i>0$, the difference between the coalition structures in $G[\bigcup_{j=1}^{i}V_{B_j}]$ and $G[\bigcup_{j=1}^{i-1}V_{B_j}]$ is $G[V_{B_i}]$. Thus, $g(c,\iso,i)-g(c,\iso,i-1)$ is 
the maximum utilitarian welfare of the coalition structure in $G[V_{B_i}]$, i.e., $\max_{k}\{f^*(B_i,\iso,k)\}$.


\paragraph{Case $\typec=\cl$}\

We define the recurrence relation for $\typec=\cl$ as follows.
\begin{align*}
g(c,\cl,i)=& \begin{cases}
0 & i = 0\\
\max_{k}\{f^*(B_1,\cl,k)\}& i = 1\\
\max\left\{
\begin{array}{c}
g(c,\cl,i-1)+\max_{k}\{f^*(B_i,\iso,k)\},\\
g(c,\iso,i-1)+\max_{k}\{f^*(B_i,\cl,k)\}
\end{array}
\right\}
& i>1.
\end{cases}
\end{align*}
The above equation holds by the following observation: 
Since $c$ forms a clique, it makes a coalition with vertices in only one of the children nodes $B_1,\ldots, B_i
$. Thus,  the role of $c$ in one of $B_1,\ldots, B_i$ is $\cl$, and the roles in the others are $\iso$. 
When $i=0$, let $g(c,\cl,0)=0$ as the base case. 
When $i=1$, $g(c,\cl,1)=\max_{k}\{f^*(B_1,\cl,k)\}$ since $\typec=\cl$. 
When $i>1$, there are two scenarios: $c$ is a member of the clique in $B_{i'}$ for some $i'<i$, or a member of the clique in $B_i$. 
In the former case, the total utilitarian welfare is $g(c,\cl,i-1)+\max_{k}\{f^*(B_i,\iso,k)\}$, and in the latter case, it is $g(c,\iso,i-1)+\max_{k}\{f^*(B_i,\cl,k)\}$.  

\paragraph{Case $\typec=\stc_\ell$}\

We define the recurrence relation for $\typec=\stc_\ell$ as follows.

As the base case, when $i=0$
\begin{align*}
g(c,\stc_\ell,0)=0.
\end{align*}

When $i>0$, we obtain
\begin{align*}
&g(c,\stc_\ell,i)\\
&= \begin{cases}
g(c,\iso,i-1)+\max_{k}\{f^*(B_i,\iso,k)\} & \ell=0\\
\max\left\{
\begin{array}{c}
g(c,\stc_\ell,i-1)+\max_{k}\{f^*(B_i,\iso,k)\},\\
g(c,\stc_{\ell-1},i-1)+\max_{k}\{f^*(B_i,\stc,k)\}-\frac{(\ell-1)(\ell+2)}{\ell(\ell+1)}
\end{array}
\right\}
& 0<\ell<i\\
-\infty & \ell> i.
\end{cases}
\end{align*}
Note that we suppose $\ell\le i$, so as an invalid case, we define $g(c,\stc_\ell,i)=-\infty$ if $\ell> i$. 

Now we see the ordinary cases. 
When $i=0$, let $g(c,\stc_\ell,0)=0$ as the base case. 
When $i>0$, we consider the difference between the coalition structures in $G[\bigcup_{j=1}^{i}V_{B_j}]$ and $G[\bigcup_{j=1}^{i-1}V_{B_j}]$. 
If $\ell=0$, $c$ is a singleton, so the maximum utilitarian welfare is $g(c,\iso,i-1)+\max_{k}\{f^*(B_i,\iso,k)\}$.
When $\ell>0$, $c$ is either in the same coalition with vertices in $B_i$ ($\max_{k}\{f^*(B_i,\stc,k)\}$) or not ($\max_{k}\{f^*(B_i,\iso,k)\}$). 
Consider merging a star with $\ell-1$ leaves and a star with a leaf to a star with $\ell$ leaves as shown in Figure \ref{fig:star:merge}.
Here, the total utility of a star with $\ell-1$ leaves plus a star with one leaf is $2(\ell-1)/\ell+1$ and the total utility of a star with $\ell$ leaves is $2\ell/(\ell+1)$. Then, the difference of total utilities before and after merging is 
\begin{align*}
  \frac{2\ell}{\ell+1}-\left(\frac{2(\ell-1)}{\ell}+1\right)
  =-\frac{(\ell-1)(\ell+2)}{\ell(\ell+1)}.
\end{align*}
\begin{figure}[htb]
\centering
\includegraphics[width=\textwidth]{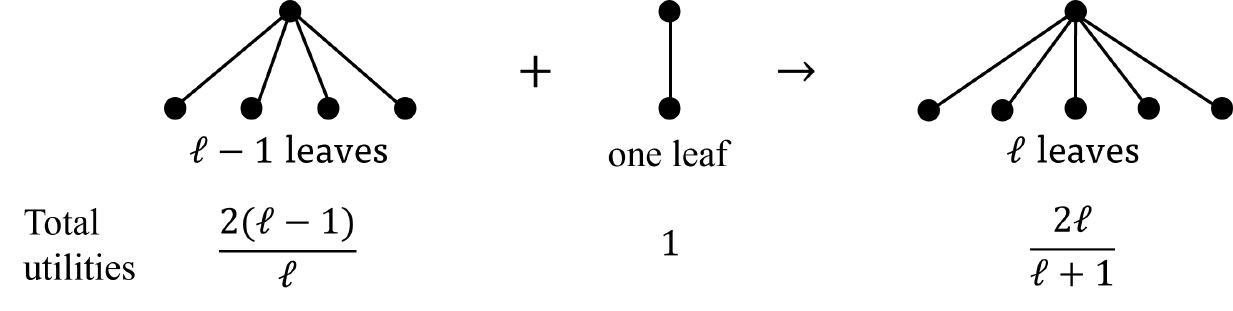}
\caption{The difference of the total utilities: a star with $\ell-1$ leaves plus a star with one leaf (before merge) and a star with $\ell$ leaves (after merge).} \label{fig:star:merge}
\end{figure}




\paragraph{Case $\typec=\stl$}\

We define the recurrence relation for $\typec=\stl$ as follows.
\begin{align*}
g(c,\stl,i)=& \begin{cases}
0 & i = 0\\
\max_{k}\{f^*(B_1,\stl,k)\} & i = 1\\ 
\max\left\{
\begin{array}{c}
g(c,\stl,i-1)+\max_{k}\{f^*(B_i,\iso,k)\},\\
g(c,\iso,i-1)+\max_{k}\{f^*(B_i,\stl,k)\}
\end{array}
\right\}
& i>1.
\end{cases}
\end{align*}
The above equation holds by the following observation: 
Since $c$ is a leaf of a star, it makes a coalition with vertices in only one of the children nodes $B_1,\ldots, B_i$. Thus, the role of $c$ in one of $B_1,\ldots, B_i$ is $\stl$, and the roles in the others are $\iso$. 
When $i=0$, let $g(c,\stl,0)=0$ as the base case. 
When $i=1$, $g(c,\stl,1)=\max_{k}\{f^*(B_1,\stl,k)\}$ since $\typec=\stl$. 
When $i>1$, $c$ joins a star coalition as a leaf in one of $B_1,\ldots,B_i$, and the role of $c$ for the other clique nodes is $\iso$. Here, there are two scenarios: $c$ is a leaf of the star in $B_{i'}$ for some $i'<i$, or a leaf of the star in $B_i$. 
In the former case, the total utilitarian welfare is $g(c,\stl,i-1)+\max_{k}\{f^*(B_i,\iso,k)\}$, and in the latter case, it is $g(c,\iso,i-1)+\max_{k}\{f^*(B_i,\stl,k)\}$.  


\subsection{Complexity}
In this subsection, we analyze the computational complexity. 
As seen in the previous subsections, our algorithm is a simple dynamic programming algorithm, though the recurrence relations are complicated. 
Since the time complexity of a dynamic programming algorithm depends on the number of states and how we evaluate recurrence relations. 
In the following, we see how many steps are required to compute each recurrence relation.

First, we consider the computational time at the block nodes. 
In a block node $B$, we compute the values of $f(B,\typeB,i,k)$ and $f_{c_{\textrm{center}}}^{v_{\textrm{leaf}}}(B,\typeB,i,k)$ for $c_{\textrm{center}}, v_{\textrm{leaf}}\in B$ in our dynamic programming algorithm. After computing all of them, we obtain $f^*(B,\typeB,k)$ in $O(\Delta)$ time. 
For a block node $B$, if we have the values of $g^*(c,\typec)$'s for each $c\in \textrm{Child}_{\cT}(B)$, we can compute each of $f(B,\typeB,i,k)$'s and $f_{c_{\textrm{center}}}^{v_{\textrm{leaf}}}(B,\typeB,i,k)$'s in a constant time. Since $|\mathcal{B}|\le n$, $|\texttt{T}_B|=O(1)$, $\max_{B\in \mathcal{B}}|\textrm{Child}_{\cT}(B)|\le \Delta$, and $0\le k\le \Delta$, all of $f(B,\typeB,i,k)$'s  can be computed in $O(n\Delta^2)$. Similarly, all of $f_{c_{\textrm{center}}}^{v_{\textrm{leaf}}}(B,\typeB,i,k)$'s can be computed in $O(n\Delta^2\cdot \Delta^2) = O(n\Delta^4)$ since there are at most $\Delta^2$ pairs of $c_{\textrm{center}}$ and $v_{\textrm{leaf}}$.

Next, we consider the computation time for $g^*(c,\typec)$ at cut node $c$. 
For $c\in \mathcal{C}$, $0\le i\le |\textrm{Child}_{\cT}(c)|\le \Delta$, and $\typec\in \{\iso,\cl,\stc_1,\ldots,\stc_{|\textrm{Child}_{\cT}(c)|},\stl\}$, we can compute $g(c,\typec,i)$ in a constant time, and hence we obtain $g^*(c,\typec) =g(c,\typec,|\textrm{Child}_{\cT}(c)|)$  in $O(\Delta)$ time.  For $\stcu$, we can compute $g^*(c,\stcu)$ in $O(\Delta)$ time by using  $g^*(c,\stc_{1}), \ldots, g^*(c,\stc_{|\textrm{Child}_{\cT}(c)|})$. 
Since $|\mathcal{C}|\le n$ and $|\texttt{T}_c| = O(\Delta)$, the total running time for $g^*$ is $O(n\Delta^2)$.


Therefore, our dynamic programming algorithm returns the maximum utilitarian welfare coalition structure in $O(n\Delta^4)$ time.

%% file: chapter05.tex
\section{Graphs of Bounded Treewidth or Vertex Cover Number}\label{chap:XP}
In this section, we design (pseudo)polynomial-time algorithms on graphs of bounded treewidth and bounded vertex cover number  for computing a maximum utilitarian welfare coalition structure and a maximum egalitarian welfare coalition structure.


\subsection{Maximizing utilitarian welfare}\label{subsec:utilitarian}

\subsubsection{Graphs with bounded treewidth}

\begin{theorem}\label{thm:treewidth:weighted}
\textsc{Utilitarian Welfare Maximization} can be computed in time $(nW)^{O(\omega)}$ where $\omega$ is the treewidth of an input graph and $W$ is the maximum absolute weight of edges.
\end{theorem}
\begin{proof}


We give an algorithm based on dynamic programming over a nice tree decomposition $\mathcal{T}=(T,\{X_t\}_{t\in V(T)})$ with $O(\omega n)$ nodes.

For each node $t$, let $\mathcal{C}_t=\{C_1,\ldots,C_{|\mathcal{C}_t|}\}$ be a partition of $X_t$. Each $C_i\in \mathcal{C}_t$ represents the set of vertices in $X_t$ that are contained in the same coalition in $G_t$.
For each $C_i\in \mathcal{C}_t$, let $D_i$ be the coalition in $G_t$ such that $C_i\subseteq D_i$. Moreover, let $\mathcal{D}_t=\{D_1,\ldots,D_{|\mathcal{D}_t|}\}$ for each node $t$.
For $D_i\in \mathcal{D}_t$, we define $n_i=|V(D_i)|$ and $m_i=|E(D_i)|$.  
For each node $t$, we store all the $\mathcal{C}_t$'s and the corresponding $n_i$ and $m_i$ for $1\le i\le |\mathcal{C}_t|$ with the maximum utilitarian welfare. In other words, we define the DP table 
$\DP[t,\mathcal{C}_t,(n_1,m_1),\ldots,(n_{|\mathcal{C}_t|},m_{|\mathcal{C}_t|})]$ 
as the maximum utilitarian welfare in the induced subgraph $G[V_t\setminus \bigcup_{D_i\in \mathcal{D}_t}D_i]$ with $n_i=|V(D_i)|$ and $m_i=|E(D_i)|$ for each $C_i\in\mathcal{D}_t$.
Since $|X_r|=0$, $\DP[r,\{\emptyset\},\emptyset]$
represents the maximum utilitarian welfare of $G$.
We recursively compute each $\DP[t,\mathcal{C}_t,(n_1,m_1),\ldots,(n_{|\mathcal{C}_t|},m_{|\mathcal{C}_t|})]$ along a given nice tree decomposition from its leaves.


\paragraph{Leaf node $t$:}
If $t$ is a leaf node, $X_t=\emptyset$. Thus, we define $\DP[t,\{\emptyset\},\emptyset]=0$ as the base case.

\paragraph{Introduce node $t$:}
For an introduce node $t$, it has a child node $t'$ such that $X_t=X_{t'}\cup\{v\}$ for some $v\in V$.
Without loss of generality, let $C_{|\mathcal{C}_t|}\in \mathcal{C}_t$ be the coalition that contains $v\in X_t$. 
Then, the recurrence relation of a DP table is defined as follows:
\begin{align*}
&\DP[t,\mathcal{C}_t,(n_1,m_1), \ldots,(n_{|\mathcal{C}_t|},m_{|\mathcal{C}_t|})]\\
&=\begin{cases}
\DP[t',\mathcal{C}_t\setminus C_{|\mathcal{C}_t|} \cup (\{C_{|\mathcal{C}_t|}\setminus \{v\}),\mathbf{m}]
& C_{|\mathcal{C}_t|}\neq\{v\}\\
\DP[t',\mathcal{C}_t\setminus \{\{v\}\},\mathbf{m'}] & C_{|\mathcal{C}_t|}=\{v\}\land  (n_{|\mathcal{C}_t|},m_{|\mathcal{C}_t|})=(1,0)\\
-\infty & \text{otherwise},
\end{cases}
\end{align*}
where\ $\mathbf{m}=((n_1,m_1),\ldots, (n_{|\mathcal{C}_t|-1},m_{|\mathcal{C}_t|-1}), (n_{|\mathcal{C}_t|}-1,m_{|\mathcal{C}_t|}-\sum_{u\in N(v)\cap C_{|\mathcal{C}_t|}}w_{vu}))$ and $\mathbf{m'}=((n_1,m_1),\ldots,(n_{|\mathcal{C}_t|},m_{|\mathcal{C}_t|}),\ldots,(n_{|\mathcal{C}_t|-1},m_{|\mathcal{C}_t|-1})).$

In the recurrence relation in an introduce node, the first case is when the coalition containing $v$ is not a singleton. In this case, $v$ is added to the existing coalition. Thus, the number of vertices and weights in the coalition increases by $1$ and the sum of weights of edges between $v$ and vertices in the coalition, respectively. Coalitions except for $C_{|\mathcal{C}_t|}$ do not change in $t$ since they are irrelevant to $v$.  
Note that the maximum utilitarian welfare in the induced subgraph $G[V_t\setminus \bigcup_{D_i\in \mathcal{D}_t}D_i]$ does not change by the definition of an introduce node.
In the second case, $v$ forms a singleton as a new coalition $C_{|\mathcal{C}_t|}=\{v\}$ in $t$. Coalitions except for $\{v\}$ do not change in $t$.  The third case represents the invalid case.

\paragraph{Forget node $t$:}
Let $t'$ be the child of $t$ and $v$ the node forgotten at $t$. Moreover, let $D_v$ be the coalition that contains $v$ in $t'$. Then, the recurrence relation of a DP table is defined as follows:
\begin{align*}
&\DP[t,\mathcal{C}_t,(n_1,m_1),\ldots,(n_{|\mathcal{C}_t|},m_{|\mathcal{C}_t|})]\\
&=
\max\left\{
\begin{array}{c}
\DP[t',\mathcal{C}_t\cup\{\{v\}\},(n_1,m_1),\ldots,(n_{|\mathcal{C}_t|},m_{|\mathcal{C}_t|}),(n_v,m_v)]+\frac{2m_v}{n_v},\\
\max_{C\in \mathcal{C}_t}{\DP[t',\mathcal{C}_t\setminus\{C\}\cup\{C\cup\{v\}\},(n_1,m_1),\ldots,(n_{|\mathcal{C}_t|},m_{|\mathcal{C}_t|})]}
\end{array} 
\right\},
\end{align*}
where $n_v = |V(D_v)|$ and $m_v = |E(D_v)|$.

The former case is when $v$ is the last vertex of its coalition, denoted by $D_v\in \mathcal{D}_{t'}$, in $t'$, and $D_v$ disappears in $t$ because $v$ is forgotten in $t$. Thus, $D_v$ never changes thereafter. Then the maximum utilitarian welfare in $G[V_t\setminus \bigcup_{D_i\in \mathcal{D}_t}D_i]$ increases by $2m_v/n_v$, which is the utilitarian welfare of $D_v$. 
The latter case is when some vertex in $D_v$ appears in  $X_t$. In this case, the maximum utilitarian welfare in $G[V_t\setminus \bigcup_{D_i\in \mathcal{D}_t}D_i]$ does not change. Since $v$ is forgotten, no edge incident to $v$ appears after node $t$, and hence $v$ can safely be removed from the coalition $D_v$.  By taking the maximum value among all the cases, we can compute the maximum utilitarian welfare in $G[V_t\setminus \bigcup_{D_i\in \mathcal{D}_t}D_i]$ with respect to $\mathcal{C}_t$ and $(n_1,m_1),\ldots,(n_{|\mathcal{C}_t|},m_{|\mathcal{C}_t|})$.


\paragraph{Join node $t$:}
Let $t_1,t_2$ be two children of a join node $t$.
At the join node, we integrate the results of $G_{t_1}$ and $G_{t_2}$. It is easily seen that the number of vertices and the sum of weights of edges in each coalition $D\in \mathcal{D}_t$ in $G_t$ can be computed by summing up them in $G_{t_1}$ and $G_{t_2}$ without double counting. Thus, the recurrence relation for join nodes is defined as follows: 
\begin{align*}
&\DP[t,\mathcal{C}_t,(n_1,m_1),\ldots,(n_{|\mathcal{C}_t|},m_{|\mathcal{C}_t|})] \\
=& 
\max\left\{\begin{array}{ll}
 & \DP[t_1,\mathcal{C}_t,(n'_1,m'_1),\ldots,(n'_{|\mathcal{C}_t|},m'_{|\mathcal{C}_t|})] \\
+& \DP[t_2,\mathcal{C}_t,(n^{\prime\prime}_1,m^{\prime\prime}_1),\ldots,(n^{\prime\prime}_{|\mathcal{C}_t|},m^{\prime\prime}_{|\mathcal{C}_t|})]
\end{array}
\right\},
\end{align*}
where $n_i=n'_i+n^{\prime\prime}_i-|C_i|$ and $m_i=m'_i+m^{\prime\prime}_i-\sum_{u,v\in C_i}w_{uv}$ for each $C_i\in \mathcal{C}_t$.

\paragraph{Complexity:}
Since $0\le n_i \le n$ and $-n^2 W \le m_i\le n^2 W$, the size of a DP table of each node is $\omega^{O(\omega)}\cdot (nW)^{O(\omega)} = (nW)^{O(\omega)}$. Each recurrence relation can be computed in time $(nW)^{O(\omega)}$. Since the number of nodes is $O(n\omega)$, the total running time is $(nW)^{O(\omega)}$.
\end{proof}

\subsubsection{Graphs of bounded vertex cover number}
\

Next, we give a polynomial-time algorithm for computing the maximum utilitarian welfare on graphs of bounded vertex cover numbers.

\begin{theorem}\label{thm:vc:weighted}
\textsc{Utilitarian Welfare Maximization} can be computed in time $n^{O(\tau)}$ where $\tau$ is the vertex cover number of an input graph. 
\end{theorem}

To prove this, we first show the key lemma as follows.
\begin{lemma}\label{lem:vc:coalitions}
For any graph $G$, there exists an optimal coalition structure with at most $\vc+1$ coalitions.
\end{lemma}
\begin{proof}
Suppose that there exists an optimal coalition structure with more than $\vc+1$ coalitions. Let $S$ be a minimum vertex cover. Then at most $\vc$ coalitions contain vertices in $S$. Thus, other coalitions only contain vertices in $V\setminus S$. Since $V\setminus S$ is an independent set, the utilitarian welfare of such a coalition is $0$. Therefore, we can merge them into one coalition without decreasing utilitarian welfare. This means that there exists an optimal coalition structure with at most $\vc+1$ coalitions.
\end{proof}


Then we design an algorithm for bounded vertex cover graphs. We first guess the number of coalitions in the optimal coalition structure. Let $b (\le \vc+1)$ be the number of coalitions. 
We further guess the number of vertices in each coalition. Let $C_1, \ldots, C_b$ be $b$ coalitions and  $n_1,\ldots, n_b$ be the number of vertices in them. Note that $n_i\ge 1$ for $1\le i\le b$. The number of possible patterns of $n_1,\ldots, n_b$ is at most $n^b$.

 Let $S$ be a minimum vertex cover of size $\vc$ in $G$. 
 We guess assignments of vertices in $S$ to $b$ coalitions. 
 The number of such possible assignments is at most $b^\vc$. If there exists a coalition $C_i$ such that the number of vertices in it exceeds $n_i$, then we immediately reject such an assignment.

 Finally, we consider assignments of vertices in $V\setminus S$ to coalitions. Suppose that the size of  coalition $C_j$ is fixed at $n_j$. Then if $v\in V\setminus S$ is assigned to $C_j$, the increase of the utilitarian welfare is computed as follows:
 \begin{align*}
    a_{vj}=\sum_{u\in N(v)\cap C_j\cap S}2w_{vu}/n_j.
\end{align*}
Note that $V\setminus S$ is an independent set. 

In order to  find a maximum utilitarian welfare coalition structure, all we need to do is to find an assignment maximizing the sum of values $a_{vj}$ for $v\in  V\setminus S$ under the capacity condition. This can be formulated as the following bin packing problem.

\begin{itembox}[l]{\textsc{Max $k$-bin packing}}
\begin{description}
    \item[Input:] $n$ items with size $1$,
    $k$ bins with capacity $c_1,\ldots,c_k (\le n)$, and
    the value $a_{ij}$ when item $i\in \{1,\ldots,n\}$ is assigned to bin $j\in \{1,\ldots,k\}$
    \item[Output:] an assignment of items to $k$ bins that maximizes the maximum value when exactly $c_1,\ldots,c_b$ items are assigned to each bottle
\end{description}
\end{itembox}
\textsc{Max $k$-bin packing} can be solved in time $n^{O(k)}$ by a simple dynamic programming algorithm.

\begin{lemma}\label{lem:binpacking}
    \textsc{Max $k$-bin packing} can be solved in time $n^{O(k)}$.
\end{lemma}

\begin{proof}
Let $\DP[i;d_1,\ldots,d_k]$ be the maximum value obtained by assigning items $1,\ldots,i$ to bins so that the capacity of bin $j$ is $d_j$ for each $j$. Then, $\DP[i;d_1,\ldots,d_k]$ can be computed by dynamic programming. First, we initialize the DP table as follows.
\begin{align*}
\DP[0;0,\ldots,0]=0.
\end{align*}
Next, we define the recurrence relation of $\DP[i;d_1,\ldots,d_k]$ as follows:
\begin{align*}
&\DP[i;d_1,\ldots,d_k]\\
&=\max
 \left\{ \max_{j\in \{1,\ldots,k\}} \{\DP[i-1;d_1,\ldots,d_{j}-1,\ldots,d_k]+a_{vj}\}\right\}.
\end{align*}
It is not hard to see the correctness of the recurrence relation. Since the size of DP table is at most $n^{k+1}$ and the recurrence relation can be computed in $O(k)$, the total running time is $O(k\cdot n^{k+1}) = n^{O(k)}$.
\end{proof}

By replacing items by vertices in $V\setminus S$ and bins by coalitions and letting $c_j=n_j - |C_j\cap S|$, we can reduce the above assignment problem to \textsc{Max $k$-bin packing}. By Lemma~\ref{lem:binpacking}, the utilitarian welfare maximization can be computed in time $\sum_{b=1}^{\tau+1}n^b\cdot b^{\vc}\cdot n^{O(b)} = n^{O(\tau)}$. Therefore, Theorem ~\ref{thm:vc:weighted} holds.


\subsection{Maximizing egalitarian welfare}
In this subsection, we consider a coalition structure that maximizes egalitarian welfare. We first design a pseudopolynomial-time algorithm on bounded treewidth graphs. 
Then we show that the egalitarian welfare maximization is NP-hard on bounded vertex cover number graphs. This is in contrast to the utilitarian welfare maximization.

\begin{theorem}\label{thm:max-min:weighted}
\textsc{Egalitarian Welfare Maximization} can be computed in time $(nW)^{O(\omega)}$ where $\omega$ is the treewidth of an input graph and $W$ is the maximum absolute weight of edges.
\end{theorem}

The algorithm in Theorem~\ref{thm:max-min:weighted} is similar to the one in Theorem~\ref{thm:treewidth:weighted} for computing the maximum utilitarian welfare coalition structure. The different point is preserving in the DP table the maximum value of the minimum utility, i.e., maximum egalitarian welfare, instead of maximum utilitarian welfare.

\bigskip

\noindent \textit{Proof of Theorem~\ref{thm:max-min:weighted}.}
For each node $t$, let $\mathcal{C}_t=\{C_1,\ldots,C_{|\mathcal{C}_t|}\}$ be a partition of $X_t$. Each $C_i\in \mathcal{C}_t$ represents the set of vertices in $X_t$ that are contained in the same coalition in $G_t$.
For each $C_i\in \mathcal{C}_t$, let $D_i$ be the coalition in $G_t$ such that $C_i\subseteq D_i$. Moreover, let $\mathcal{D}_t=\{D_1,\ldots,D_{|\mathcal{D}_t|}\}$ for each node $t$.
Let $X_t=\{v_1,\ldots,v_{|X_t|}\}$ and $D_v$ denote the coalition that contains $v$ in $G_t$.
For $D_i\in \mathcal{D}_t$, we define $n_i=|V(D_i)|$, $u_j=\sum_{u\in N(v_j)\cap D_{v_j}}w_{v_ju}$ and $b_i=\min_{v\in D_i\setminus C_i}\{\sum_{u\in N(v)\cap D_v}w_{vu}\}$.  

For each node $t$, we store all the $\mathcal{C}_t$'s and the corresponding $n_i$, $u_j$ and $b_i$ for $1\le i\le |\mathcal{C}_t|$ and $1\le j\le |X_t|$ with the maximum egalitarian welfare. In other words, we define the DP table 
$\DP[t,\mathcal{C}_t,n_1,\ldots,$ $n_{|\mathcal{C}_t|},u_1,\ldots,u_{|X_t|},b_1,\ldots,b_{|\mathcal{C}_t|}]$ 
as the maximum egalitarian  welfare in the induced subgraph $G[V_t\setminus \bigcup_{D_i\in \mathcal{D}_t}D_i]$ with $n_i=|V(D_i)|$, $u_j=\sum_{u\in N(v_j)\cap D_{v_j}}w_{v_ju}$ and $b_i=\min_{v\in D_i\setminus C_i}\{\sum_{u\in N(v)\cap D_v}w_{vu}\}$ for each $C_i\in\mathcal{D}_t$ and $v_j\in X_t$.
Since $|X_r|=0$, $\DP[r,\{\emptyset\},\emptyset,\emptyset,\emptyset]$
represents the maximum egalitarian welfare of $G$.
We recursively compute each $\DP[t,\mathcal{C}_t,n_1,\ldots,$ $n_{|\mathcal{C}|_t},u_1,\ldots,u_{|X_t|},b_1,\ldots,b_{|\mathcal{C}|_t}]$ along a given nice tree decomposition from its leaves.

\paragraph{Leaf node $t$:}
If $t$ is a leaf node, $X_t=\emptyset$. Thus we define $\DP[t,\{\emptyset\},\emptyset,\emptyset,\emptyset]=0$ as the base case.

\paragraph{Introduce node $t$:}
For an introduce node $t$, it has a child node $t'$ such that $X_t=X_{t'}\cup\{v\}$ for some $v\in V$.
Without loss of generality, let $C_{|\mathcal{C}_t|}\in \mathcal{C}_t$ be the coalition that contains $v\in X_t$ and $v=v_{|X_t|}$. 
Let $\mathbf{n}=(n_1,\ldots,n_{|\mathcal{C}_t|}-1)$, $\mathbf{u}=(u'_1,\ldots,u'_{|X_t|-1})$ where $u'_i=u_i-w _{v_iv}$ if $v_i\in C_{|\mathcal{C}_t|}$ and $u'_i=u_i$ otherwise, $\mathbf{b}=(b_1,\ldots,b_{|\mathcal{C}_t|})$,
$\mathbf{n'}=(n_1,\ldots,n_{|\mathcal{C}_t|-1})$, $\mathbf{u'}=(u_1,\ldots,u_{|X_t|-1})$ and $\mathbf{b'}=(b_1,\ldots,b_{|\mathcal{C}_t|-1})$.
Then, the recurrence relation of a DP table is defined as follows:

\begin{flalign}
&\DP[t,\mathcal{C}_t,n_1,\ldots,n_{|\mathcal{C}_t|},u_1,\ldots,u_{|X_t|},b_1,\ldots,b_{|\mathcal{C}_t|}]&\nonumber
\end{flalign}
\begin{empheq}
[left = {
= \empheqlbrace \,}]{alignat=2}
& \DP[t',\mathcal{C}_t\setminus C_{|\mathcal{C}_t|} \cup (\{C_{|\mathcal{C}_t|}\setminus \{v\}\}),\mathbf{n},\mathbf{u},\mathbf{b}]     &   &    \left\{  \begin{aligned}
 & C_{|\mathcal{C}_t|}\neq\{v\},  \\
 & u_{|X_t|}=\sum_{u\in N(v)\cap C_{|\mathcal{C}_t|}}w_{vu} \nonumber
\end{aligned} \right.\\
& \DP[t',\mathcal{C}_t\setminus \{\{v\}\},\mathbf{n'},\mathbf{u'},\mathbf{b'}]     &     &    
\left\{\begin{aligned}
        & C_{|\mathcal{C}_t|}=\{v\},   \\
        & (n_{|\mathcal{C}_t|},u_{|X_t|},b_{|\mathcal{C}_t|})=(1,0,\infty) \nonumber
\end{aligned} \right.\\
& -\infty   &\quad &    \text{otherwise.} \nonumber
\end{empheq}


In the recurrence relation in an introduce node, the first case is when the coalition containing $v$ is not a singleton. In this case, $v$ is added to the existing coalition. Thus, the number of vertices and the sum of weights of edges of each vertex with its neighbors in the coalition increases by $1$ and the sum of weights of edges between $v$ and vertices in the coalition, respectively. 
Note that the set of vertices $D_t\setminus C_{|\mathcal{C}_t|}$, does not change by the definition of an introduce node.
Coalitions except for $C_{|\mathcal{C}_t|}$ do not change in $t$ since they are irrelevant to $v$.  
In the second case, $v$ forms a singleton as a new coalition $C_{|\mathcal{C}_t|}=\{v\}$ in $t$. Coalitions except for $\{v\}$ do not change in $t$.  The third case represents the invalid case.

\paragraph{Forget node $t$:}
Let $t'$ be the child of $t$ and $v$ the node forgotten at $t$. Moreover, let $D_v$ be the coalition that contains $v$ in $t'$ and $C_v$ be the set of vertices that contains $v$ in $X_{t'}$.
Then, the recurrence relation of a DP table is defined as follows:
\begin{align*}
&\DP[t,\mathcal{C}_t,n_1,\ldots,n_{|\mathcal{C}_t|},u_1,\ldots,u_{|X_t|},b_1,\ldots,b_i,\ldots,b_{|\mathcal{C}_t|}]\\
&=
\max\left\{
\begin{array}{c}
\min\{\DP[t',\mathcal{C}_t\cup\{\{v\}\},\mathbf{n^{\prime\prime}},\mathbf{u^{\prime\prime}},\mathbf{b^{\prime\prime}}],\frac{\min\{b'_v,u_v\}}{n_v}\},\\
\max\{\DP[t',\mathcal{C}_t\setminus\{C_i\}\cup\{C_i\cup\{v\}\},\mathbf{n^{\prime\prime\prime}},\mathbf{u^{\prime\prime}},\mathbf{b^{\prime\prime\prime}}]\mid C_i\in\mathcal{C}_t:u_v\ge b_i\}
\end{array} 
\right\},
\end{align*}
where $u_v=\sum_{u\in N(v)\cap D_v}w_{vu}$, $n_v=|D_v|$, $b'_v=\min_{u\in D_v\setminus C_v}\{\sum_{w\in N(u)\cap D_v}w_{uw}\}$, $\mathbf{n^{\prime\prime}}=(n_1,\ldots,n_{|\mathcal{C}_t|},n_v)$,
\ $\mathbf{u^{\prime\prime}}=(u_1,\ldots,u_{|X_t|},u_v)$,\ $\mathbf{b^{\prime\prime}}=(b_1,\ldots,b_{|\mathcal{C}_t|},b'_v)$,
\\ $\mathbf{n^{\prime\prime\prime}}=(n_1,\ldots,n_{|\mathcal{C}_t|})$,
\ and $\mathbf{b^{\prime\prime\prime}}=(b_1,\ldots,b_i,\ldots,b_{|\mathcal{C}_t|})$.

The former case is when $v$ is the last vertex of its coalition, denoted by $D_v\in \mathcal{D}_{t'}$, in $t'$, and $D_v$ disappears in $t$ because $v$ is forgotten in $t$. Thus, $D_v$ never changes thereafter. Then the maximum egalitarian welfare in $G[V_t\setminus \bigcup_{D_i\in \mathcal{D}_t}D_i]$ can be computed by taking the minimum value of the maximum egalitarian welfare in $G[V_{t'}\setminus \bigcup_{D_i\in \mathcal{D}_{t'}}D_i]$ and $\frac{\min\{b'_v,u_v\}}{n_v}$, which is the minimum utility of $D_v$. 
The latter case is when some vertex in $D_v$ appears in $X_t$. In this case, the maximum egalitarian welfare in $G[V_t\setminus \bigcup_{D_i\in \mathcal{D}_t}D_i]$ does not change. Since $v$ is forgotten, no edge incident to $v$ appears after node $t$, and hence $v$ can safely be removed from the coalition $D_v$. 
Note that $u_v$ must be not less than $b_i$ (otherwise, among each vertex in $D_v\setminus X_t$ in $G_t$, the minimum value of the sum of weights of edges with adjacent vertices in the coalition would be $u_v$). 
By taking the maximum value among all the cases, we can compute the maximum egalitarian welfare in $G[V_t\setminus \bigcup_{D_i\in \mathcal{D}_t}D_i]$ with respect to $\mathcal{C}_t$ and $n_1,\ldots,n_{|\mathcal{C}_t|},u_1,\ldots,u_{|X_t|},b_1,\ldots,b_i,\ldots,b_{|\mathcal{C}_t|}$.

\paragraph{Join node $t$:}
Let $t_1,t_2$ be two children of a join node $t$.
At the join node, we integrate the results of $G_{t_1}$ and $G_{t_2}$. It is easily seen that the number of vertices in each coalition $D\in \mathcal{D}_t$ in $G_t$ and the sum of weights of edges with adjacent vertices belonging to the same coalition in $G_t$ for each vertex can be computed by summing up them in $G_{t_1}$ and $G_{t_2}$ without double counting. Moreover, for each $D_i\setminus C_i$ in $G_t$, the minimum value of the sum of weights of edges among vertices in same coalition can be computed by taking the minimum value of them in $G_{t_1}$ and $G_{t_2}$ respectively. Thus, the recurrence relation for join nodes is defined as follows: 
\begin{align*}
&\DP[t,\mathcal{C}_t,n_1,\ldots,n_{|\mathcal{C}_t|},u_1,\ldots,u_{|X_t|},b_1,\ldots,b_{|\mathcal{C}_t|}]\\
&=
\max\left\{
\begin{array}{c}
\min\{\DP[t_1,\mathcal{C}_t,n'_1,\ldots,n'_{|\mathcal{C}_t|},u'_1,\ldots,u'_{|X_t|},b'_1,\ldots,b'_{|\mathcal{C}_t|}],\\
\DP[t_2,\mathcal{C}_t,n^{\prime\prime}_1,\ldots,n^{\prime\prime}_{|\mathcal{C}_t|},u^{\prime\prime}_1,\ldots,u^{\prime\prime}_{|X_t|},b^{\prime\prime}_1,\ldots,b^{\prime\prime}_{|\mathcal{C}_t|}]\}
\end{array} 
\right\},
\end{align*}
where $n_i=n'_i+n^{\prime\prime}_i-|C_i|$, $u_j=u'_j+u^{\prime\prime}_j-\sum_{u\in C_i\cap N(v_j)}w_{uv_j}$, and $b_i=\min\{b'_i,b^{\prime\prime}_i\}$ for each $C_i\in \mathcal{C}_t$ and $v_j\in X_t$.

\paragraph{Complexity:}
Since $0\le n_i \le n$, $-nW \le u_j\le nW$ and $-nW \le b_i\le nW$, the size of a DP table of each node is $\omega^{O(\omega)}\cdot (nW)^{O(\omega)} =(nW)^{O(\omega)}$. Each recurrence relation can be computed in time $n^{O(\omega)}$. Since the number of nodes is $O(n\omega)$, the total running time is $(nW)^{O(\omega)}$. \qed
\bigskip

Since $\omega\le \vc$ holds for any graph, we have the following corollary.
\begin{corollary}\label{cor:Egal:unweighted:vc}
\textsc{Egalitarian Welfare Maximization} can be computed in time $(nW)^{O(\vc)}$ where $\vc$ is the vertex cover number of an input graph and $W$ is the maximum absolute weight of edges.
\end{corollary}


On the other hand, even restricting $\vc=4$, we can show that \textsc{Egalitarian Welfare Maximization} is NP-hard.

\begin{theorem}\label{thm:Egal:vc}
\textsc{Egalitarian Welfare Maximization} is NP-hard on graphs of bounded vertex cover number. 
\end{theorem}
\begin{figure}
   \centering
   \includegraphics[width=0.7\textwidth]{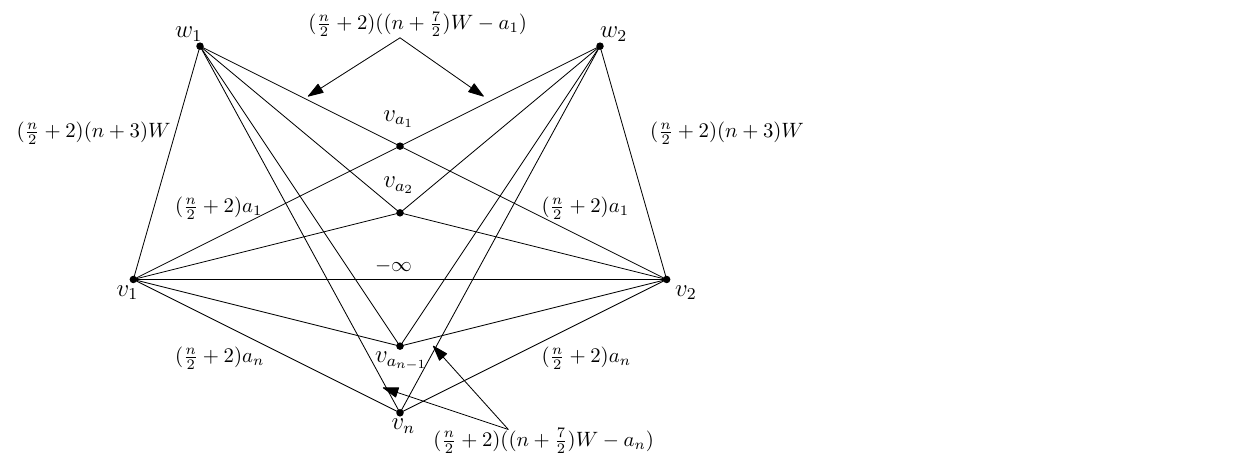}
   \caption{The constructed graph $G$.}
   \label{fig:thm:Egal:vc}
\end{figure}
\begin{proof}
We give a reduction from \textsc{Partition}.
In the problem, given a set $A=\{a_1,\ldots,a_n\}$ of $n$ integers, the task is to determine whether there exists a subset $A'\subseteq A$ such that $\sum_{a\in A'}=W/2$ where $W=\sum_{a\in A}a$.
This problem is NP-hard if we require $|A'|=n/2$~\cite{DBLP:books/fm/GareyJ79}. Thus, we suppose that $|A|$ is even.

From an instance of \textsc{Partition}, we construct an instance of \textsc{Egalitarian Welfare Maximization} (see Figure \ref{fig:thm:Egal:vc}).
First, we create four vertices $v_1,v_2,w_1,w_2$ and $n$ vertices $v_{a_1},\ldots,v_{a_n}$ corresponding to $n$ integers. Let $V_A=\{v_{a_1},\ldots,v_{a_n}\}$. 
Two vertices $v_1$ and $v_2$ are connected to each $v_{a_i}\in V_A$ by edges $\{v_1,v_a\}$ and $\{v_2,v_a\}$ with weight $(n/2+2)a_i$, respectively.
For $w_1$ and $w_2$, we connect them to  $v_{a_i}\in V_A$ by edges with weight $(n/2+2)((n+7/2)W-a_i)$. Moreover, we add edges $\{v_1,w_1\}$ and  $\{v_2,w_2\}$ with weight $(n/2+2)(n+3)W$.
Finally, we connect $v_1$ and $v_2$ by a large negative weight edge. The weight $-\infty$ denotes a large negative weight.
Let $G=(V,E)$ be the constructed graph.
Then $G$ has a vertex cover $\{v_1, v_2, w_1,w_2\}$ of size $4$ because the graph obtained by deleting them forms singletons.

We show that there exists a subset $A'\subseteq A$ such that $\sum_{a\in A'}=W/2$ if and only if there exists a partition $\mathcal{P}$ of $V$ such that the least utility among agents is at least  $(n+7/2)W$.

Suppose that we are given $A_1\subseteq A$ such that $\sum_{a\in A_1}=W/2$. Let $A_2 = A\setminus A_1$. Then we denote the corresponding vertex set to $A_1$ and $A_2$ of $V$ by $V_{A_1}$ and $V_{A_2}$, respectively. 
Let $C_1=V_{A_1}\cup\{v_1,w_1\}$ and $C_2=V_{A_2}\cup\{v_2,w_2\}$. Note that $|C_1|=|C_2|=n/2+2$.
Define a partition $\mathcal{P}=\{C_1, C_2\}$.
We show that $\mathcal{P}$ is a partition  of $V$ such that the least utility among agents is at least  $(n+7/2)W$.
For $w_1$, the utility is at least $(n+3)W +(n+7/2)W - W/2>(n+7/2)W$. Similarly, the utility of $w_2$ is at least $(n+7/2)W$.
For each $v_{a_i}\in V_A$, the utility is $(n+7/2)W$.
Finally, the utility of $v_1$ (resp., $v_2$) is $(n+3)W+W/2=(n+7/2)W$.

Conversely, we are given a partition $\mathcal{P}$ of $V$ such that the least utility among agents is at least  $(n+7/2)W$.
Since the edge between $v_1$ and $v_2$ has large negative weight, they belong to different coalitions in $\mathcal{P}$, otherwise, the utility is less than $(n+7/2)W$.

If $v_i$ and $w_i$ belong to different coalition for $i\in \{1,2\}$, the utility of  $v_i$ is at most $(n+2)W< (n+7/2)W$.
Thus, $v_i$ and $w_i$ belong to the same coalition. Let $C_i$ be the coalition containing $v_i$ and $w_i$ for each $i$. 
Next, if there exists a vertex $v_{a}\in V_A$ not belonging to both $C_1$ and $C_2$, its utility becomes $0$. Therefore, each vertex in $V_A$ belongs to either $C_1$ or $C_2$ and $\mathcal{P}$ consists of two coalitions $C_1$ and $C_2$.

We then show that $|C_1|=|C_2|=n/2+2$. If not, without loss of generality, we suppose that $|C_1|>n/2+2$.
For a vertex $v_{a_i}\in V_A$, the utility is at most $\left((n/2+2)((n+7/2)W-a_i)+(n/2+2)a_i)\right)/(n/2+3)=\left((n/2+2)(n+7/2)W\right)\\/(n/2+3)<(n+7/2)W$. Thus, we have $|C_1|=|C_2|=n/2+2$.

Let $V_{A_1}=C_1\cap V_A$ and $V_{A_2}=C_2\cap V_A$. Moreover, $A_1$ and $A_2$ denotes the subsets of $A$ corresponding to $V_{A_1}$ and $V_{A_2}$, respectively. Note that $|A_1|=|A_2|=n/2$.
Then the sum of weights of edges incident to $v_1$ in $C_1$ is  $(n/2+2)\sum_{a\in A_1} + (n/2+2)(n+3)W$.
Suppose that the utility of $v_1$ is more than $(n+7/2)W$.
As $|C_1|=n/2+2$ and the weight of $\{v_1,w_1\}$ is $(n/2+2)(n+3)W$, we have  $\sum_{a\in A_1}>W/2$. Then  $\sum_{a\in A_2} < W/2$. This implies that $u(v_2)$ is less than $(n+7/2)W$. 
Therefore, the utility of $v_1$ is exactly $(n+7/2)W$. Eventually, $\sum_{a\in A_1}a = \sum_{a\in A_2}a = W/2$.
This completes the proof.
\end{proof}